%% file: main_arxiv.tex
\documentclass{scrartcl}

\usepackage[utf8]{inputenc}
\usepackage[T1]{fontenc}
\usepackage{microtype}
\usepackage{lmodern}

\usepackage{graphicx}
\usepackage{xcolor}
\usepackage[fleqn]{amsmath}
\usepackage{amsfonts}
\usepackage{amssymb}
\usepackage{amsthm}
\usepackage{mathtools}
\usepackage{stmaryrd}
\usepackage{bm}
\usepackage{colortbl}
\usepackage{multirow}
\usepackage{import}
\usepackage{marvosym}

\usepackage{macros}
\usepackage{tabularx}
\usepackage[strings]{underscore}

\usepackage{authblk}
\usepackage{thm-restate}

\newtheorem{definition}{Definition}
\newtheorem{theorem}   {Theorem}
\newtheorem{proposition} [theorem]{Proposition}
\newtheorem{lemma}       [theorem]{Lemma}

\definecolor{niceredbright}{HTML}{bd0310}
\definecolor{nicebluebright}{HTML}{197b9b}
\definecolor{nicered}{HTML}{7f0a13}
\definecolor{niceblue}{HTML}{104354}
\definecolor{nicegreen}{HTML}{217516}
\definecolor{nicepurple}{HTML}{884bab}
\definecolor{nicebg}{HTML}{f6f0e4}
\definecolor{niceredlight}{HTML}{c9888d}
\definecolor{nicebluelight}{HTML}{78a4b8}
\definecolor{nicegreenlight}{HTML}{76de68}
\definecolor{nicepurplelight}{HTML}{bc87db}

\makeatletter
\RequirePackage[bookmarks,unicode,colorlinks=true]{hyperref}%
   \def\@citecolor{niceblue}%
   \def\@urlcolor{niceblue}%
   \def\@linkcolor{nicered}%

\def\orcidID#1{\smash{\href{http://orcid.org/#1}{\protect\raisebox{-1.25pt}{\protect\includegraphics{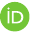}}}}}
\makeatother

\usepackage{etoolbox}
\makeatletter
\pretocmd\start@gather{%
    \if@minipage\kern-\topskip\kern-\baselineskip\kern+7pt\fi
}{}{}
\makeatother

\newcommand{\Proofsketch}{Proof sketch}

\renewcommand{\parag}{\paragraph*}

\begin{document}
\title{Running Time Analysis of Broadcast Consensus Protocols\thanks{This work was supported by an ERC Advanced Grant (787367: PaVeS) and by the Research Training Network of the Deutsche Forschungsgemeinschaft (DFG) (378803395: ConVeY).}}

\author{Philipp Czerner \orcidID{0000-0002-1786-9592}, Stefan Jaax \orcidID{0000-0001-5789-8091}}
\affil{\{czerner, jaax\}@in.tum.de\\
Department of Informatics, TU München, Germany}
\maketitle              
\begin{abstract}
Broadcast consensus protocols (BCPs) are a model of computation, in which anonymous, identical, finite-state agents compute by sending/receiving global broadcasts. BCPs are known to compute all number predicates in $\NL=\NSPACE(\log n)$ where $n$ is the number of agents. They can be considered an extension of the well-established model of population protocols.
This paper investigates execution time characteristics of BCPs.
We show that every predicate computable by population protocols is computable by a BCP with expected $\mathcal{O}(n \log n)$ interactions, which is asymptotically optimal. We further show that every log-space, randomized Turing machine can be simulated by a BCP with $\mathcal{O}(n \log n \cdot T)$ interactions in expectation, where $T$ is the expected runtime of the Turing machine. This allows us to characterise polynomial-time BCPs as computing exactly the number predicates in $\mathsf{ZPL}$, i.e.\ predicates decidable by log-space bounded randomised Turing machine with zero-error in expected polynomial time where the input is encoded as unary.
\end{abstract}

\section{Introduction}\label{sec:introduction}
\input{sec-intro}

\section{Preliminaries}\label{sec:prelims}
\input{sec-prelims}

\section{Example: Majority}\label{sec:exmp}
\input{sec-example}

\section{Comparison with other Models}\label{sec:cmp}
\input{sec-comparison}

\subsection{Non-Deterministic Broadcast Protocols}\label{sec:nondeterministic}
\input{ssec-detbroadcast}

\subsection{Population Protocols}
\input{ssec-pop-protocols}

\section{Protocols for Presburger Arithmetic}\label{sec:fast-presburger}
\input{sec-fast-presburger}

\section{Protocols for all Predicates in ZPL}\label{sec:fast-general}
\input{sec-fast-general}

\bibliographystyle{splncs04}
\bibliography{references}

\appendix
\input{appendix.tex}

\end{document}

%% file: sec-intro.tex
In recent years, models of distributed computation following the \emph{computation-by-consensus} paradigm attracted considerable interest in research (see for example \cite{aspnes2009introduction,michail2011mediated,michail2015terminating,aspnes2017clocked,concur19}). In such models, network agents compute number predicates, i.e.\ Boolean-valued functions of the type $\N^k \rightarrow \{0, 1\}$, by reaching a stable consensus whose value determines the outcome of the computation.
Perhaps the most prominent model following this paradigm are \emph{population protocols} \cite{angluin2006computation,angluin2006stably}, a model in which anonymous, identical, finite-state agents interact randomly in pairwise rendezvous to agree on a common Boolean output.

Due to anonymity and locality of interactions, it is an inherent property of population protocols that agents are generally unable to detect with absolute certainty when the computation has stabilized. This makes sequential composition of protocols difficult, and further complicates the implementation of control structures such as loops or branching statements.
To overcome this drawback, two kinds of approaches have been suggested in the literature: 1.) Let agents guess when the computation has stabilized, leading to composable, but merely \emph{approximately correct} protocols \cite{angluin2008fast,popfast}, or 2.) extend population protocols by global communication primitives that enable agents to query global properties of the agent population \cite{concur19,aspnes2017clocked,michail2015terminating}.

Approaches of the first kind are for the most part based on simulations of global broadcasts
by means of \emph{epidemics}. In epidemics-based approaches the spread of the broadcast signal is simulated by random pairwise rendezvous, akin to the spread of a viral epidemic in a population. When the broadcasting agent meets a certain fraction of ``infected'' agents, it may decide with reasonable certainty that the broadcast has propagated throughout the entire population, which then leads to the initiation of the next computation phase. Of course, the decision to start the next phase may be premature, in which case the rest of the execution may be faulty. However, epidemics can also be used to implement phase clocks that help keep the failure probability low (see e.g.\ \cite{angluin2008fast}).

In \cite{concur19}, Blondin, Esparza, and one of the authors of this paper introduced  \emph{broadcast consensus protocols} (BCPs), an extension of population protocols by reliable, global, and atomic broadcasts. BCPs find their precursor in the broadcast protocol model introduced by Emerson and Namjoshi in~\cite{EmersonN98} to describe bus-based hardware protocols. This model has been investigated intensely in the literature, see e.g.~\cite{EsparzaFM99,FinkelL02,DelzannoRB02,SchmitzS13}. Broadcasts also arise naturally in biological systems. For example, Uhlendorf \textit{et al.}\  analyse applications of broadcasts in the form of an external, global light source for controlling a population of yeasts ~\cite{BDGG17}.

The authors of \cite{concur19} show that BCPs compute precisely the predicates in $\NL = \NSPACE(\log n)$, where $n$ is the number of agents. For comparison, it is known that population protocols compute precisely the \emph{Presburger predicates}, which are the predicates definable in the first-order theory of the integers with addition and the usual order -- a class of predicates much less expressive than the former.

An epidemics-based approach was used in \cite{angluin2008fast} to show that population protocols can simulate with high probability a step of a virtual register machine with expected $\mathcal{O}(n \log^5(n))$ interactions, where $n$ is the number of agents. This result stimulated further research into time bounds for classical problems such as leader election (see e.g.\ \cite{gkasieniec2018fast,alistarh2017time,doty2018stable,sudo2020leader,berenbrink2020optimal}) and majority (see e.g.\ \cite{alistarh2015fast,alistarh2018space}). In their seminal paper \cite{angluin2006computation}, Angluin \emph{et al.}\ already showed that population protocols can stably compute Presburger predicates with $\mathcal{O}(n^2 \log n)$ interactions in expectation. Belleville  \emph{et al.} further showed that leaderless protocols require a quadratic number of interactions in expectation to stabilize to the correct output for a wide class of predicates \cite{belleville2017hardness}. The aforementioned bounds apply to \emph{stabilisation time}: the time it takes to go from an initial configuration to a stable consensus that cannot be destroyed by future interactions. In \cite{popfast}, Kosowski and Uznanski considered the weaker notion of \emph{convergence time}: the time it takes on average to ultimately transition to the correct consensus (although this consensus could in principle be destroyed by future interactions), and they show that sublinear convergence time is achievable.

By contrast, to the best of our knowledge, time characteristics of BCPs have not been discussed in the literature. The $\NL$-powerful result presented in \cite{concur19} does not establish any time bounds. In fact, \cite{concur19} only considers a non-probabilistic variant of BCPs with some global fairness assumption in place of probabilistic choices.

\parag{Contributions of the paper}
This paper initiates the runtime analysis of BCPs in terms of expected number of interactions to reach a stable consensus. To simplify the definition of probabilistic execution semantics, we introduce a restricted, deterministic variant of BCPs without rendezvous transitions. In Section~\ref{sec:prelims}, we define probabilistic execution semantics for the restricted version of BCPs, and we provide an introductory example for a fast protocol computing majority in Section~\ref{sec:exmp}.

In Section~\ref{sec:cmp}, we show that these restrictions of our BCP model are inconsequential in terms of expected number of interactions: both rendezvous and nondeterministic choices can be simulated with a constant runtime overhead.

In Section~\ref{sec:fast-presburger}, we show that every Presburger predicate can be computed by BCPs with $\mathcal{O}(n \log n)$ interactions and with constant space, where $n$ denotes the number of agents in the population. This result is asymptotically optimal.

In more generality, in Section~\ref{sec:fast-general}, we use BCPs to simulate Turing machines (TMs). In particular, we show that any randomised, logarithmically space-bound, polynomial-time TM can be simulated by a BCP with an overhead of $\mathcal{O}(n\log n)$ interactions per step. Conversely, any polynomial-time BCP can be simulated by such a TM. This result can be considered an improvement of the $\NL$ bound from \cite{concur19}, now in a probabilistic setting. We also give a corresponding upper bound, which yields the following succinct characterisation: polynomial-time BCPs compute exactly the number predicates in $\mathsf{ZPL}$, which are the languages decidable by randomised log-space polynomial-time TMs with zero-error (the log-space analogue to $\mathsf{ZPP}$).

Bounding the time requires a careful analysis of each step in the simulation of the Turing machine. Thus, our proof diverges in significant ways from the proof establishing the $\NL$ lower bound in \cite{concur19}. Most notably, we now make use of epidemics in order to implement clocks that help reduce failure rates.

%% file: sec-prelims.tex
\newcommand{\Op}[1]{\mathsf{#1}}

\parag{Complexity classes} As is usual, we define $\NL$ as the class of languages decidable by a nondeterministic log-space TM. Additionally, by $\mathsf{ZPL}$ we denote the set of languages decided by a randomised log-space TM $A$, s.t.\ $A$ only terminates with the correct result (zero-error) and that it terminates within $\mathcal{O}(\operatorname{poly}n)$ steps in expectation, as defined by Nisan in \cite{nisan1993read}.

\parag{Multisets} A \emph{multiset} over a finite set $E$ is a mapping $M \colon E \to \N$. The set of all multisets over $E$ is denoted $\N^E$. For every $e \in E$, $M(e)$ denotes the number of occurrences of $e$ in $M$. We sometimes denote multisets using a set-like notation, \eg $\multiset{f, g, g}$ is the multiset $M$ such that $M(f) = 1$, $M(g) = 2$ and $M(e) = 0$ for every $e \in E \setminus \{f, g\}$.  Addition, comparison and scalar multiplication are extended to multisets componentwise, i.e.\ $(M \mplus M')(e) \defeq M(e) + M'(e)$, $(\lambda M)(e)\defeq\lambda M(e)$ and $M \leq M' \defiff M(e) \leq M'(e)$ for every $M,M'\in\mathbb{N}^Q$, $e \in E$, and $\lambda\in\mathbb{N}$. For $M'\le M$ we also define componentwise subtraction, i.e.\ $(M-M')(e)\defeq M(e)-M'(e)$ for every $e \in E$. For every $e \in E$, we write $\vec{e} \defeq \multiset{e}$. We lift functions $f \colon E\rightarrow E'$ to multisets by defining $f(M)(e')\defeq\sum_{f(e)=e'}M(e)$ for $e'\in E'$. Finally, we define the \emph{support} and \emph{size} of $M \in \N^E$ respectively as $\supp{M} \defeq \{e \in E : M(e) > 0\}$ and $|M| \defeq \sum_{e \in E} M(e)$.

\parag{Broadcast Consensus Protocols} A \emph{broadcast consensus protocol} \cite{concur19} (BCP) is a tuple $\mathcal{P}=(Q,\Sigma,\delta,I,O)$ where
\begin{itemize}
  \item $Q$ is a non-empty, finite set of \emph{states},
  \item $\Sigma$ is a non-empty, finite \emph{input alphabet},
  \item $\delta$ is the \emph{transition function} (defined below),
  \item $I \colon \Sigma \rightarrow Q$ is the \emph{input mapping}, and
  \item $O \subseteq Q$ is a set of \emph{accepting states}.
\end{itemize}

The function $\delta$ maps every state $q \in Q$ to a pair $(r, f)$ consisting of the \emph{successor state} $r \in Q$ and the \emph{response function} $f \colon Q \trans[0pt]{} Q$.

\parag{Configurations}
A \emph{configuration} is a multiset $C\in\N^Q$. Intuitively, a configuration $C$ describes a collection of identical finite-state \emph{agents} with $Q$ as set of states, containing $C(q)$ agents in state $q$ for every
$q \in Q$. We say that $C\in\N^Q$ is a \emph{$1$-consensus} if $\supp{C}\subseteq O$, and a \emph{$0$-consensus} if $\supp{C}\subseteq Q\setminus O$.

\parag{Step relation}
A broadcast $\delta(q)=(r,f)$ is executed in three steps: (1)~an agent at state $q$ broadcasts a signal and leaves $q$; (2)~all other agents receive the signal and move to the states indicated by the function $f$, i.e.\ an agent in state $s$ moves to $f(s)$; and (3)~the broadcasting agent enters state $r$.

Formally, for two configurations $C, C'$ we write $C\trans[0pt]{}C'$, whenever there exists a state $q\in Q$ s.t.\ $C(q)\ge 1$, $\delta(q)=(r,f)$, and $C'=f(C-\vec{q})+\vec{r}$ is the configuration computed from $C$ by the above three steps. By $\trans{*}$ we denote the reflexive-transitive closure of $\trans[0pt]{}$.

For example, consider a configuration $C\defeq\multiset{a,a,b}$ and a broadcast transition $a\mapsto b,\{a\mapsto c, b\mapsto d\}$. To execute this transition, we move an agent from state $a$ to state $b$ and apply the transition function to all other agents, so we end up in $C'\defeq\multiset{b}+\multiset{c,d}$.

\parag{Broadcast transitions}
We write broadcast transitions as $q\mapsto r, S$ with $S$ a set of expressions $q'\mapsto r'$. This refers to $\delta(q)=(r,f)$, with $f(q')=r'$ for $(q'\mapsto r')\in S$. We usually omit identity mappings $q'\mapsto q'$ when specifying $S$.

For graphic representations of broadcast protocols we use a different notation, which separates sending and receiving broadcasts. There we identify a transition $\delta(q)=(r,f)$ with a name $\alpha$ and specify it by writing $q\trans{!\alpha}r$ and $q'\trans{?\alpha}r'$ for $f(q')= r'$. Intuitively, $q'\trans{?\alpha}r'$ can be understood as an agent transitioning from $q'$ to $r'$ upon receiving the signal $\alpha$, and $q\trans{!\alpha}r$ means that an agent in state $q$ may transmit the signal $\alpha$ and simultaneously transition to state $r$.

As defined, $\delta$ is a total function, so each state is associated with a unique broadcast. If we do not specify a transition $\delta(q)=(r,f)$ explicitly, we assume that it simply maps each state to itself, i.e.\ $q\mapsto q,\{r\mapsto r:r\in Q\}$. We refer to those transitions as \emph{silent}.

\parag{Executions}
An \emph{execution} is an infinite sequence $\pi = C_0 C_1 C_2...$ of configurations with $C_i\trans[0pt]{}C_{i+1}$ for every $i$. It has some fixed number of agents $n\defeq|C_0|=|C_1|=...$ . Given a BCP and an initial configuration $C_0\in\N^Q$, we generate a random execution with the following Markov chain: to perform a step at configuration $C_i$, a state $q\in Q$ is picked at random with probability distribution $p(q)=C_i(q)/|C_i|$, and the (uniquely defined) transition $\delta(q)$ is executed, giving the successor configuration $C_{i+1}$. We refer to the random variable corresponding to the trace of this Markov chain as \emph{random execution}.

\parag{Stable Computation}
Let $\pi$ denote an execution and $\infmany(\pi)$ the configurations occurring infinitely often in $\pi$. If $\infmany(\pi)$ contains only $b$-consensuses, we say that $\pi$ \emph{stabilises} to $b$. For a predicate $\varphi:\N^\Sigma\rightarrow\{0,1\}$ we say that $\PP$ \emph{(stably) computes} $\varphi$, if for all inputs $X\in\N^\Sigma$, the random execution of $\PP$ with initial configuration $C_0=I(X)$ stabilises to $\varphi(X)$ with probability $1$.

Finally, for an execution $\pi=C_0 C_1 C_2...$ we let $T_\pi$ denote the smallest $i$ s.t.\ all configurations in $C_iC_{i+1}...$ are $\varphi(X)$-consensuses, or $\infty$ if no such $i$ exists. We say that a BCP $\mathcal{P}$ \textit{computes $\varphi$ within $f(n)$ interactions}, if for all initial configurations $C_0$ with $n$ agents the random execution $\pi$ starting at $C_0$ has $\mathbb{E}(T_\pi)\le f(n)<\infty$, i.e.\ $\mathcal{P}$ stabilises within $f(n)$ steps in expectation.
\footnote{Note that this implies that the BCP computes $\varphi$, as else there is a positive probability that $T_\pi=\infty$, and thus $\mathbb{E}(T_\pi)=\infty$ as well.}
If $f\in\mathcal{O}(\operatorname{poly}(n))$, then we call $\mathcal{P}$ a \emph{polynomial-time} BCP.

\parag{Global States}
Often, it is convenient to have a shared global state between all agents. If, for a BCP $\PP=(Q,\Sigma,\delta,I,O)$ we have $Q=S\times G$, $I(\Sigma)\subseteq Q\times\{j\}$ for some $j\in G$, and $f((s,j))\in Q\times\{j'\}$ for each $\delta((q,j))=((r,j'),f)$, then we say that $\PP$ has \emph{global states $G$}. A configuration $C$ has \emph{global state $j$}, if $\supp{C}\subseteq Q\times\{j\}$ for $j\in G$. Note that, starting from a configuration with global state $j$, $\PP$ can only reach configurations with a global state. Hence for $\PP$ we will generally only consider configurations with a global state. To make our notation more concise, when specifying a transition $\delta(q)=(r,f)$ for $\PP$, we will write $f$ as a mapping from $S$ to $S$, as $q,r$ already determine the mapping of global states.

\parag{Population Protocols} A population protocol \cite{angluin2006computation} replaces broadcasts by local rendezvous. It can be specified as a tuple $(Q, \Sigma, \delta, I, O)$ where $Q$, $\Sigma$, $I$, $O$ are defined as in BCPs, and $\delta \colon Q^2\rightarrow Q^2$ defines \emph{rendezvous transitions}. A step of the protocol at $C$ is made by picking two agents uniformly at random, and applying $\delta$ to their states: first $q_1 \in Q$ is picked with probability $C(q_1)/|C|$, then $q_2 \in Q$ is picked with probability $C'(q_2)/|C'|$, where $C' \defeq C-\multiset{q_1}$. The successor configuration then is $C - \multiset{q_1, q_2} + \multiset{r_1, r_2}$ where $\delta(q_1, q_2) = (r_1, r_2)$.

\parag{Broadcast Protocols} Later on we will construct BCPs out of smaller building blocks which we call \emph{broadcast protocols (BPs)}. A BP is a pair $(Q,\delta)$, where $Q$ and $\delta$ are defined as for BCPs. We extend the applicable definitions from above to BPs, in particular the notions of configurations, executions, and global states.

%% file: sec-example.tex
\newcommand{\frozenx}{\tilde{x}}
\newcommand{\activey}{\tilde{y}}
\newcommand{\send}[1]{!{#1}}
\newcommand{\rec}[1]{?{#1}}
\newcommand{\q}[1]{#1}

\newcommand{\Sx}{\alpha}
\newcommand{\Sy}{\beta}

\begin{figure}
\def\svgwidth{122mm}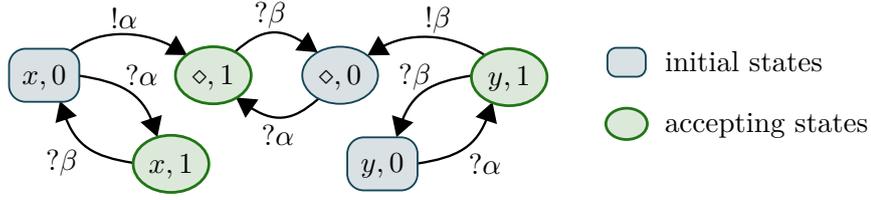
\caption{A fast broadcast consensus protocol computing the majority predicate.}
\label{fig:ppMajority}
\end{figure}

As an introductory example, we construct a broadcast consensus protocol for the \emph{majority predicate} $\varphi(x,y)=x>y$. Figure~\ref{fig:ppMajority} depicts the protocol graphically. We have the set of states $\{x,y,\diamond\}\times\{0,1\}$, with global states $\{0,1\}$, where the states $O\defeq\{(x,1),(y,1),(\diamond,1)\}$ are accepting, and $I(x)=(x,0)$ and $I(y)=(y,0)$. The transitions are
\begin{align}
&(x,0)\mapsto(\diamond,1),\,\emptyset \tag{$\alpha$} \\
&(y,1)\mapsto(\diamond,0),\,\emptyset \tag{$\beta$}
\end{align}

Note that we use the more compact notation for transitions in the presence of global states, written in long form ($\alpha$) would be
\begin{align}
&(x,0)\mapsto(\diamond,1),\,\{(x,0)\mapsto(x,1),(y,0)\mapsto(y,1),(\diamond,0)\mapsto(\diamond,1)\} \tag{$\alpha$}
\end{align}

To make the presentation of the following sample execution more readable, we shorten the state $(i,j)$ to $i_j$. For input $x=3$ and $y=2$, an execution could look like this:
\begin{alignat*}{6}
&\multiset{x_0,x_0,x_0,y_0,y_0}
&&\trans{\Sx} \multiset{\diamond_1,x_1,x_1,y_1,y_1}
&&\trans{\Sy} \multiset{\diamond_0,x_0,x_0,\diamond_0,y_0} \\
\trans{\Sx}\; &\multiset{\diamond_1,\diamond_1,x_1,\diamond_1,y_1}
&&\trans{\Sy} \multiset{\diamond_0,\diamond_0,x_0,\diamond_0,\diamond_0}
&&\trans{\Sx} \multiset{\diamond_1,\diamond_1,\diamond_1,\diamond_1,\diamond_1}
\end{alignat*}

Intuitively, there is a preliminary global consensus, which is stored in the global state. Initially, it is rejecting, as $x>y$ is false in the case $x=y=0$. However, any $x$ agent is enough to tip the balance, moving to an accepting global state. Now any $y$ agent could speak up, flipping the consensus again.

The two factions initially belonging to $x$ and $y$, respectively, alternate in this manner by sending signals $\Sx$ and $\Sy$. Strict alternation is ensured as an agent will not broadcast to confirm the global consensus, only to change it.

After emitting the signal, the agent from the corresponding faction goes into state $\diamond$, where it can no longer influence the computation. In the end, the majority faction remains and determines the final consensus.

Considering these alternations with shrinking factions, the expected number of steps of the protocol until stabilization can be bounded by $2\sum_{k=1}^{n}{n}/{k}=\mathcal{O}(n\log n)$. To see that this holds, we consider the factions separately: let $n_0$ denote the number of agents the first faction starts with (i.e.\ agents initially in state $(x,0)$), and $n_1$ the number at the end. When we are waiting for the first transition of this faction all $n_0$ agents are enabled, so we wait $n/n_0$ steps in expectation until one of them executes a broadcast. For the next one, we wait $n/(n_0-1)$ steps. In total, this yields $\sum_{k=n_1+1}^{n_0}n/k\le\sum_{k=1}^nn/k$ steps for the first faction, and via the same analysis for the second as well.

In contrast to the $\mathcal{O}(n\log n)$ interactions this protocol takes, constant-state population protocols require $n^2$ interactions in expectation for the computation of majority \cite{alistarh2015fast}.
However, these numbers are not directly comparable: broadcasts may not be parallelizable, while it is uncontroversial to assume that $n$ rendez-vous occur in parallel time $1$.

%% file: figures/majoritycomms.pdf_tex
\begingroup%
  \makeatletter%
  \providecommand\color[2][]{%
    \errmessage{(Inkscape) Color is used for the text in Inkscape, but the package 'color.sty' is not loaded}%
    \renewcommand\color[2][]{}%
  }%
  \providecommand\transparent[1]{%
    \errmessage{(Inkscape) Transparency is used (non-zero) for the text in Inkscape, but the package 'transparent.sty' is not loaded}%
    \renewcommand\transparent[1]{}%
  }%
  \providecommand\rotatebox[2]{#2}%
  \newcommand*\fsize{\dimexpr\f@size pt\relax}%
  \newcommand*\lineheight[1]{\fontsize{\fsize}{#1\fsize}\selectfont}%
  \ifx\svgwidth\undefined%
    \setlength{\unitlength}{345.82677165bp}%
    \ifx\svgscale\undefined%
      \relax%
    \else%
      \setlength{\unitlength}{\unitlength * \real{\svgscale}}%
    \fi%
  \else%
    \setlength{\unitlength}{\svgwidth}%
  \fi%
  \global\let\svgwidth\undefined%
  \global\let\svgscale\undefined%
  \makeatother%
  \begin{picture}(1,0.22131148)%
    \lineheight{1}%
    \setlength\tabcolsep{0pt}%
    \put(0,0){\includegraphics[width=\unitlength,page=1]{majoritycomms.pdf}}%
    \put(0.11301494,0.12312215){\color[rgb]{0,0,0}\makebox(0,0)[t]{\lineheight{1.25}\smash{\begin{tabular}[t]{c}$x,0$\end{tabular}}}}%
    \put(0.24979804,0.0268528){\color[rgb]{0,0,0}\makebox(0,0)[t]{\lineheight{1.25}\smash{\begin{tabular}[t]{c}$x,1$\end{tabular}}}}%
    \put(0.29539236,0.12312215){\color[rgb]{0,0,0}\makebox(0,0)[t]{\lineheight{1.25}\smash{\begin{tabular}[t]{c}$\diamond,1$\end{tabular}}}}%
    \put(0.43217565,0.12312215){\color[rgb]{0,0,0}\makebox(0,0)[t]{\lineheight{1.25}\smash{\begin{tabular}[t]{c}$\diamond,0$\end{tabular}}}}%
    \put(0,0){\includegraphics[width=\unitlength,page=2]{majoritycomms.pdf}}%
    \put(0.47776991,0.02783442){\color[rgb]{0,0,0}\makebox(0,0)[t]{\lineheight{1.25}\smash{\begin{tabular}[t]{c}$y,0$\end{tabular}}}}%
    \put(0.19899396,0.18421205){\color[rgb]{0,0,0}\makebox(0,0)[t]{\lineheight{1.25}\smash{\begin{tabular}[t]{c}$!\Sx$\end{tabular}}}}%
    \put(0.21794841,0.12187465){\color[rgb]{0,0,0}\makebox(0,0)[t]{\lineheight{1.25}\smash{\begin{tabular}[t]{c}$?\Sx$\end{tabular}}}}%
    \put(0.13189985,0.03003019){\color[rgb]{0,0,0}\makebox(0,0)[t]{\lineheight{1.25}\smash{\begin{tabular}[t]{c}$?\Sy$\end{tabular}}}}%
    \put(0.36515902,0.0576848){\color[rgb]{0,0,0}\makebox(0,0)[t]{\lineheight{1.25}\smash{\begin{tabular}[t]{c}$?\Sx$\end{tabular}}}}%
    \put(0.35671693,0.19053671){\color[rgb]{0,0,0}\makebox(0,0)[t]{\lineheight{1.25}\smash{\begin{tabular}[t]{c}$?\Sy$\end{tabular}}}}%
    \put(0.53751827,0.18559551){\color[rgb]{0,0,0}\makebox(0,0)[t]{\lineheight{1.25}\smash{\begin{tabular}[t]{c}$!\Sy$\end{tabular}}}}%
    \put(0.51270123,0.12307074){\color[rgb]{0,0,0}\makebox(0,0)[t]{\lineheight{1.25}\smash{\begin{tabular}[t]{c}$?\Sy$\end{tabular}}}}%
    \put(0.58834149,0.02612411){\color[rgb]{0,0,0}\makebox(0,0)[t]{\lineheight{1.25}\smash{\begin{tabular}[t]{c}$?\Sx$\end{tabular}}}}%
    \put(0.61455301,0.12214046){\color[rgb]{0,0,0}\makebox(0,0)[t]{\lineheight{1.25}\smash{\begin{tabular}[t]{c}$y,1$\end{tabular}}}}%
    \put(0.78245841,0.13780724){\color[rgb]{0,0,0}\makebox(0,0)[lt]{\lineheight{1.25}\smash{\begin{tabular}[t]{l}initial states\end{tabular}}}}%
    \put(0.78245841,0.07127995){\color[rgb]{0,0,0}\makebox(0,0)[lt]{\lineheight{1.25}\smash{\begin{tabular}[t]{l}accepting states\end{tabular}}}}%
    \put(0,0){\includegraphics[width=\unitlength,page=3]{majoritycomms.pdf}}%
  \end{picture}%
\endgroup%

%% file: sec-comparison.tex
To facilitate the definition of an execution model, we only consider deterministic BCPs, in the sense that for each state there is a unique transition to execute. Blondin, Esparza and Jaax~\cite{blondin2019expressive} analysed a more general model, i.e.\ they allow multiple transitions for a single state, picking one of them uniformly at random when an agent in that state sends a broadcast. Additionally, as they consider BCPs as an extension of population protocols, they include rendez-vous transitions.

We now show that we can simulate both extensions within a constant-factor overhead.

%% file: ssec-detbroadcast.tex
The following construction allows for two broadcast transitions to be executed uniformly at random from a single state. This can easily be extended to any constant number of transitions using the usual construction of a binary tree with rejection sampling.

Now assume that we are given a BCP $(Q, \Sigma, \delta_0, I, F)$ with another set of broadcast transitions $\delta_1$ and we want each agent to pick one transition uniformly at random from $\delta_0$ or $\delta_1$ whenever it executes a broadcast.

We implement this using a synthetic coin, i.e.\ we are utilising randomness provided by the scheduler to enable individual agents to make random choices. This idea has also been used for population protocols~\cite{alistarh2017time,alistarh2018recent}. Compared to these implementations, broadcasts allow for a simpler approach.

The idea is that we partition the agents into types, so that half of the agents have type 0 and the other half have type 1. Additionally, there is a global coin shared across all agents. To flip the coin, a random agent announces its type (the coin is set to heads if the agent is type 0, tails if it is type 1) and a second random agent executes a broadcast transition from either $\delta_0$ or $\delta_1$, depending on the state of the global coin that has just been set. These two steps repeat, the former flipping the coin fairly and the latter then executing the actual transitions. Figure~\ref{fig:multbcast} sketches this procedure.

Intuitively, we start with no agents having either type 0 or 1. When such a typeless agent is picked by the scheduler to announce its type (to flip the global coin) it instead broadcasts that it is searching for a partner. Once this has happened twice, these two agents are matched, one is assigned type 0 and the other type 1. Thus we ensure that there is the exact same number of type 0 and type 1 agents at all times, meaning that we get a perfectly fair coin. Additionally we make progress regardless of whether an agent with or without a type is chosen.

To describe the construction formally, we introduce a set of types $T\defeq\{?,+,-,0,1\}$, and choose the set of states $Q'\defeq Q\times T\times \{*,0,1\}$, with global states $\{*,0,1\}$ used to represent the state of the synthetic coin. We use $(q,?)$ as initial state instead of $q\in I$, and start with global state $*$. To pick types, we need transitions
\begin{alignat}{2}
&(q,?,*)\mapsto (q,+,*), \{(r,?)\mapsto (r,-):r\in Q\} \qquad &&\text{for $q\in Q$} \tag{seek}\\
&
\begin{aligned}
(q,-,*)\mapsto (q,1,*),\,&\{(r,-)\mapsto (r,?):r\in Q\} \\
\cup\,&\{(r,+)\mapsto (r,0):r\in Q\}
\end{aligned}
\qquad &&\text{for $q\in Q$} \tag{find}
\end{alignat}

\begin{figure}[t]
\def\svgwidth{122mm}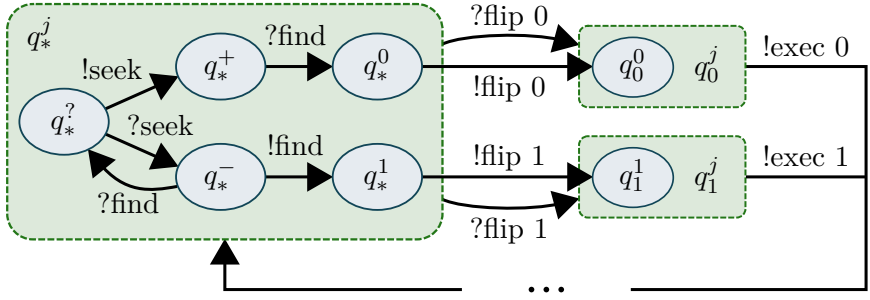
\caption{Transition diagram for implementing multiple broadcasts per state, for $q\in Q$, with $(q,i,j)$ written as $q^i_j$. Dashed nodes represent multiple states, with $j\in T$. Transitions resulting from executing the broadcasts in $\delta_0,\delta_1$ are not shown.}
\label{fig:multbcast}
\end{figure}

So an agent of type $?$ announces that it seeks a partner, moving itself to type $+$ and the others to type $-$. Then any type $-$ agent may broadcast that a match has been found, moving itself to type $1$ and the type $+$ agent to type $0$. The other type $-$ agents revert to type $?$. This ensures that the number of type $0$ and $1$ agents is always equal. Note that there may be an odd number of agents, in which case one agent of type $+$ remains.

The following transitions effectively flip the global coin, by having an agent of type $0$ or $1$ announce that we now execute a broadcast transition from respectively $\delta_0$ or $\delta_1$. Here, we have $q\in Q, \circ\in \{0,1\}$. 
\begin{align}
& (q,\circ,*)\mapsto (q,\circ,\circ),\,\emptyset \tag{flip $\circ$}
\end{align}
Then we actually execute the transition $\delta_\circ(q)=(r,f)$, for each $(q,i)\in Q\times T$.
\begin{align}
& (q,i,\circ)\mapsto (r,i,*), \{(s,j)\mapsto (f(s),j): (s,j)\in Q\times T\} \tag{exec $\circ$}
\end{align}

As the number of type $0$ and $1$ agents is equal, we select transitions from $\delta_0$ and $\delta_1$ uniformly at random. It remains to show that the overhead of this scheme is bounded.

Executing transition (exec 0) or (exec 1) is the goal. Transitions (flip 0) and (flip 1) ensure that the former are executed in the very next step, so they cause at most a constant-factor slowdown. Transitions (seek) and (find) can be executed at most $n$ times, as they decrease the number of agent of type $?$. All that remains is the implicit silent transition of states $(q,+,j)$, which occurs with probability at most $1/n$ in each step.

Hence, to execute $m\ge n$ steps of the simulated protocol our construction takes at most $(2m+2n)\cdot n/(n-1)\le 8m$ steps in expectation.

%% file: figures/multbcast2.pdf_tex
\begingroup%
  \makeatletter%
  \providecommand\color[2][]{%
    \errmessage{(Inkscape) Color is used for the text in Inkscape, but the package 'color.sty' is not loaded}%
    \renewcommand\color[2][]{}%
  }%
  \providecommand\transparent[1]{%
    \errmessage{(Inkscape) Transparency is used (non-zero) for the text in Inkscape, but the package 'transparent.sty' is not loaded}%
    \renewcommand\transparent[1]{}%
  }%
  \providecommand\rotatebox[2]{#2}%
  \newcommand*\fsize{\dimexpr\f@size pt\relax}%
  \newcommand*\lineheight[1]{\fontsize{\fsize}{#1\fsize}\selectfont}%
  \ifx\svgwidth\undefined%
    \setlength{\unitlength}{345.82677165bp}%
    \ifx\svgscale\undefined%
      \relax%
    \else%
      \setlength{\unitlength}{\unitlength * \real{\svgscale}}%
    \fi%
  \else%
    \setlength{\unitlength}{\svgwidth}%
  \fi%
  \global\let\svgwidth\undefined%
  \global\let\svgscale\undefined%
  \makeatother%
  \begin{picture}(1,0.32786885)%
    \lineheight{1}%
    \setlength\tabcolsep{0pt}%
    \put(0,0){\includegraphics[width=\unitlength,page=1]{multbcast2.pdf}}%
    \put(0.10066159,0.1813241){\color[rgb]{0,0,0}\makebox(0,0)[t]{\lineheight{1.25}\smash{\begin{tabular}[t]{c}$q^?_*$\end{tabular}}}}%
    \put(0.26959519,0.12136122){\color[rgb]{0,0,0}\makebox(0,0)[t]{\lineheight{1.25}\smash{\begin{tabular}[t]{c}$q^-_*$\end{tabular}}}}%
    \put(0.79458168,0.12038894){\color[rgb]{0,0,0}\makebox(0,0)[t]{\lineheight{1.25}\smash{\begin{tabular}[t]{c}$q^j_1$\end{tabular}}}}%
    \put(0.07495838,0.27272537){\color[rgb]{0,0,0}\makebox(0,0)[t]{\lineheight{1.25}\smash{\begin{tabular}[t]{c}$q^j_*$\end{tabular}}}}%
    \put(0.79458168,0.24031464){\color[rgb]{0,0,0}\makebox(0,0)[t]{\lineheight{1.25}\smash{\begin{tabular}[t]{c}$q^j_0$\end{tabular}}}}%
    \put(0.26959519,0.24128692){\color[rgb]{0,0,0}\makebox(0,0)[t]{\lineheight{1.25}\smash{\begin{tabular}[t]{c}$q^+_*$\end{tabular}}}}%
    \put(0.43852878,0.24128692){\color[rgb]{0,0,0}\makebox(0,0)[t]{\lineheight{1.25}\smash{\begin{tabular}[t]{c}$q^0_*$\end{tabular}}}}%
    \put(0.71415241,0.12172027){\color[rgb]{0,0,0}\makebox(0,0)[t]{\lineheight{1.25}\smash{\begin{tabular}[t]{c}$q^1_1$\end{tabular}}}}%
    \put(0.43852878,0.12136122){\color[rgb]{0,0,0}\makebox(0,0)[t]{\lineheight{1.25}\smash{\begin{tabular}[t]{c}$q^1_*$\end{tabular}}}}%
    \put(0.71415241,0.24164597){\color[rgb]{0,0,0}\makebox(0,0)[t]{\lineheight{1.25}\smash{\begin{tabular}[t]{c}$q^0_0$\end{tabular}}}}%
    \put(0.15239573,0.23591088){\color[rgb]{0,0,0}\makebox(0,0)[t]{\lineheight{1.25}\smash{\begin{tabular}[t]{c}!seek\end{tabular}}}}%
    \put(0.20469766,0.17497657){\color[rgb]{0,0,0}\makebox(0,0)[t]{\lineheight{1.25}\smash{\begin{tabular}[t]{c}?seek\end{tabular}}}}%
    \put(0.35043226,0.27297925){\color[rgb]{0,0,0}\makebox(0,0)[t]{\lineheight{1.25}\smash{\begin{tabular}[t]{c}?find\end{tabular}}}}%
    \put(0.17016824,0.08966852){\color[rgb]{0,0,0}\makebox(0,0)[t]{\lineheight{1.25}\smash{\begin{tabular}[t]{c}?find\end{tabular}}}}%
    \put(0.58114862,0.21712281){\color[rgb]{0,0,0}\makebox(0,0)[t]{\lineheight{1.25}\smash{\begin{tabular}[t]{c}!flip 0\end{tabular}}}}%
    \put(0.58114862,0.14247825){\color[rgb]{0,0,0}\makebox(0,0)[t]{\lineheight{1.25}\smash{\begin{tabular}[t]{c}!flip 1\end{tabular}}}}%
    \put(0.58114862,0.29391112){\color[rgb]{0,0,0}\makebox(0,0)[t]{\lineheight{1.25}\smash{\begin{tabular}[t]{c}?flip 0\end{tabular}}}}%
    \put(0.58114862,0.06295123){\color[rgb]{0,0,0}\makebox(0,0)[t]{\lineheight{1.25}\smash{\begin{tabular}[t]{c}?flip 1\end{tabular}}}}%
    \put(0.90123305,0.26079239){\color[rgb]{0,0,0}\makebox(0,0)[t]{\lineheight{1.25}\smash{\begin{tabular}[t]{c}!exec 0\end{tabular}}}}%
    \put(0.90123305,0.14086667){\color[rgb]{0,0,0}\makebox(0,0)[t]{\lineheight{1.25}\smash{\begin{tabular}[t]{c}!exec 1\end{tabular}}}}%
    \put(0.35043226,0.15305357){\color[rgb]{0,0,0}\makebox(0,0)[t]{\lineheight{1.25}\smash{\begin{tabular}[t]{c}!find\end{tabular}}}}%
  \end{picture}%
\endgroup%

%% file: ssec-pop-protocols.tex
Another extension to BCPs is the addition of rendez-vous transitions. Here we are given a map $R:Q^2\rightarrow Q^2$. At each step, we flip a coin and either execute a broadcast transition as usual, or pick two distinct agents uniformly at random, in state $q$ and $r$, respectively. These interact and move to the two states $R(q,r)$.

Again, we can simulate this extension with only a constant-factor increase in the expected number of steps. Given a BCP $(Q, \Sigma, B, I, F)$, the idea is to add states $\{\tilde{q}:q\in Q\}\cup\{r_q:r,q\in Q\}$ and insert “activating” transitions $q\mapsto \tilde{q},\{r\mapsto r_q:r\in Q\}$ for $q\in Q$ and “deactivating” transitions $r_q\mapsto s,\{\tilde{q}\mapsto t\}\cup\{u_q\mapsto u:u\in Q\}$ for each $R(q,r)=(s,t)$. So a state $q$ first signals that it wants to start a rendez-vous transition. Then, any other state $r$ answers, both executing the transition and signalling to all other states that it has occurred.

Each state in $Q$ has exactly 2 broadcast transitions, so (using the scheme described above) the probability of executing any “activating” transition is exactly $\frac{1}{2}$, the same as doing one of the original broadcast transitions in $B$. After doing an activating transition we may do nothing for a few steps by executing the broadcast transition on $\tilde{q}$, but eventually we execute a “deactivating” transition and go back. The probability of executing a broadcast on $\tilde{q}$ is $1/n$, so simulating a single rendez-vous transition takes $1+n/(n-1)\le 3$ steps in expectation.

%% file: sec-fast-presburger.tex
While Blondin, Esparza and Jaax~\cite{blondin2019expressive} show that BCPs are more expressive than population protocols, they leave the question open whether BCPs provide a runtime speed-up for the class of Presburger predicates computable by population protocols.
We already saw that Majority can be computed within $\mathcal{O}(n \log n)$ interactions in BCPs. This also holds in general for Presburger predicates:

\begin{theorem}
  Every Presburger predicate is computable by a BCP within at most $\mathcal{O}(n \log n)$ interactions.
\end{theorem}

We remark that the $\mathcal{O}(n \log n)$ bound is asymptotically optimal: e.g.\ the stable consensus for the parity predicate $(x = 1 \text{ mod } 2)$ must alternate with configuration size, which clearly requires every agent to perform at least one broadcast in the computation, and thus yields a lower bound of $\sum_{k=1}^{n} \frac{n}{k} = \Omega(n \log n)$ steps like in the coupon collector's problem \cite{flajolet1992birthday}.

It is known \cite{ginsburg1966} that every Presburger predicate can be expressed  as Boolean combination of linear inequalities and linear congruence equations over the integers, i.e. as Boolean combination of predicates of the form $\sum_i \alpha_ix_i < c$, and $\sum_i \alpha_i x_i = c \text{ mod } m$, where the $\alpha_i$, $c$ and $m$ are integer constants.
In Section~\ref{sec:linearinequalities} we will construct BCPs that compute arbitrary linear inequalities, before we briefly sketch the construction for congruences and Boolean combinations in Section~\ref{sec:moduloandbool}.

\subsection{Linear Inequalities}\label{sec:linearinequalities}
\begin{proposition}\label{prop:linearinequalities}
Let $\alpha_1, \ldots, \alpha_k, c\in\mathbb{Z}$ and let $\varphi(x_1, \ldots, x_k) \defiff \sum_{i=1}^k \alpha_i x_i < c$ denote a linear inequality. There exists a broadcast consensus protocol that computes $\varphi$ within $\mathcal{O}(n \log n)$ interactions in expectation.
\end{proposition}
\begin{proof}
We assume wlog that $\alpha_i\ne 0$ for $i=1,...,k$ and that $\alpha_1,...,\alpha_k$ are pairwise distinct.

Let $A \defeq \text{max}\{|\alpha_1|, |\alpha_2 |\ldots, |\alpha_k|, |c|\}$. We define a BCP $\PP = (Q\times G, \Sigma, \delta, I, O)$ with global states $G$, where
\begin{align*}
Q & \defeq  \{0, \alpha_1,...,\alpha_k\} &
\Sigma & \defeq \{x_1, \ldots, x_k \} \\
G & \defeq [-2A, 2A] &
O & \defeq \{(q,v):v<c\}
\end{align*}
As inputs we get $I(x_i)\defeq(\alpha_i,0)$ for each $i=1,...,k$.
The transitions $\delta$ are constructed as follows. For every $v \in [-2A, 2A]$ and every $\alpha_i$ satisfying $v+\alpha_i \in [-2A, 2A]$, we add the following transition to $T$:
\begin{equation}
(\alpha_i,v)\mapsto(0,v+\alpha_i),\,\emptyset \tag{$\alpha_i$}
\end{equation}

Intuitively, in the first component of its state an agent stores its contribution to $\sum_i\alpha_ix_i$, the left-hand side of the inequality. The global state is used to store a counter value, initially set to $0$. Each agent adds its contribution to the counter, as long as it does not overflow. The counter goes from $-2A$ to $2A$, which allows it to store the threshold plus any single contribution. The final counter value then determines the outcome of the computation.

\parag{Correctness}
Let $\Op{ctr}(C)$ denote the global state (and thus current counter value) of configuration $C$. Further, let
\[\Op{sum}(C) \defeq \sum_{(\alpha, v) \in Q} C(\alpha,v) \cdot \alpha + \Op{ctr}(C)\]
denote the sum of all agents' contributions and the current value of the counter. Every initial configuration $C_0$ has $\Op{ctr}(C)=0$ and thus $\Op{sum}(C)=\sum_i\alpha x_i$. Each transition $\alpha$ increases the counter by $\alpha$ but sets the agent's contribution to $0$ (from $\alpha$), so $\Op{sum}(C)$ is constant throughout the execution.

  Recall that our output mapping depends only on the value of the counter, so our agents always form a consensus (though not necessarily a stable one). If this consensus and $\varphi(C_0)$ disagree, then, we claim, a non-silent transition is enabled.

  To see this, note that the current consensus depends on whether $\Op{ctr}(C)<c$. If that is the case, but $\varphi(C_0)=0$, then $\Op{sum}(C)\ge c$ and some agent with positive contribution $\alpha>0$ exists. Due to $\Op{ctr}(C)<c$, transition $\alpha$ is enabled. Conversely, if $\Op{ctr}(C)\ge c$ and $\varphi(C_0)=1$, some transition $\alpha$ with $\alpha<0$ will be enabled.

  Finally, note that each non-silent transition increases the number of agents with contribution $0$ by one, so at most $n$ can be executed in total. So the execution converges and reaches, by the above argument, a correct consensus.

  \parag{Convergence time}
  Each agent executes at most one non-silent transition. To estimate the total number of steps, we partition the agents by their current contribution: for a configuration $C$ let $C^+\defeq C\upharpoonright\{(q,v)\in Q:q>0\}$ denote the agents with positive contribution, and define $C^-$ analogously. We have that either $\Op{ctr}(C)<0$ and all transitions of agents in $C^+$ would be enabled, or $\Op{ctr}(C)\ge 0$ and the transitions of $C^-$ could be executed.

  If $C^+$ is enabled, then we have to wait at most $n/|C^+|$ steps in expectation until a transition is executed, which reduces $|C^+|$ by one. In total we get $n/|C_0^+|+n/(|C_0^+|-1)+...+n/1\in\mathcal{O}(n\log n)$. The same holds for $C^-$, yielding our overall bound of $\mathcal{O}(n\log n)$.
\end{proof}

\subsection{Modulo Predicates and Boolean Combinations}\label{sec:moduloandbool}

\begin{proposition}
  Let $\varphi(x_1,...,x_k) \defiff \sum_{i=1}^k\alpha_ix_i \equiv c\pmod l < c$ denote a linear inequality, with $\alpha_1, \ldots, \alpha_k, c,l\in\mathbb{Z}, l\ge2$. There exists a broadcast consensus protocol that computes $\varphi$ within $\mathcal{O}(n \log n)$ interactions in expectation.
\end{proposition}
\begin{proof}[\Proofsketch]
The idea is the same as for Proposition~\ref{prop:linearinequalities}, but instead of taking care not to overflow the counter we simply perform the additions modulo $l$.
\end{proof}

\begin{restatable}[Boolean combination of predicates]{proposition}{restatebppforboolean} \label{prop:bppforboolean}
Let $\varphi$ be a Boolean combination of predicates $\varphi_1,...,\varphi_k$, which are computed by BCPs $\PP_1,...,\PP_k$, respectively, within $\mathcal{O}(n\log n)$ interactions. Then there is a protocol computing $\varphi$ within $\mathcal{O}(n\log n)$ interactions.
\end{restatable}
\begin{proof}[\Proofsketch]
We do a simple parallel composition of the $k$ BCPs, which is the same construction as used for ordinary population protocols (see for example \cite[Lemma 6]{angluin2006computation}). The full proof can be found in Appendix~\ref{app:bppforboolean}.
\end{proof}

%% file: sec-fast-general.tex

BCPs compute precisely the predicates in $\mathsf{NL}$ with input encoded in unary, which corresponds to $\mathsf{NSPACE}(n)$ when encoded in binary. The proof of the $\NL$ lower bound by Blondin, Esparza and Jaax~\cite{blondin2019expressive} goes through multiple stages of reduction and thus does not reveal which predicates can be computed \emph{efficiently}. We will now take a more direct approach, using a construction similar to the one by Angluin, Aspnes and Eisenstat~\cite{angluin2008fast}. A step of a randomised Turing machine (RTM) can be simulated using variants of the protocols for Presburger predicates from Section~\ref{sec:fast-presburger}, which we combine with a clock to determine whether the step has finished, with high probability.

Instead of simulating RTMs directly, it is more convenient to first reduce them to counter machines. Here, we will use counter machines that are both randomised and capable of multiplying and dividing by two, with the latter also determining the remainder. This ensures that the reduction is performed efficiently, i.e.\ with overhead of $\mathcal{O}(n\log n)$ interactions per step. Our counter machines do not support a decrement operation, as this is not needed for our reduction (it could also easily be simulated using the other operations).

We start by showing the other direction, that BCPs can be simulated efficiently by RTMs.

\begin{lemma}\label{lem:polybcpinzl}
Polynomial-time BCPs compute at most the predicates in $\mathsf{ZPL}$ with input encoded in unary.
\end{lemma}
\begin{proof}
An RTM can store the number of agents in each state as binary counters. Picking an agent uniformly at random can be done in $\mathcal{O}(\log n)$ time by picking a random number between 1 and $n$ and comparing it to the agents in the different states. Simulating a transition can also be done with logarithmic overhead. It can further be shown (see Appendix~\ref{app:proof-stable}) that stabilization of the execution is decidable in time $\mathcal{O}(\log n)$. As the BCP uses only $\mathcal{O}(\operatorname{poly}n)$ interactions (in expectation) the RTM is also $\mathcal{O}(\operatorname{poly}n)$ time-bounded.
\end{proof}

\begin{theorem}\label{thm:polybppiszl}
Polynomial-time BCPs compute exactly the predicates in $\mathsf{ZPL}$ with input encoded in unary.
\end{theorem}

The proof of Theorem~\ref{thm:polybppiszl} will take up the remainder of this section. But before we get into the details of that, we need to formally define counter machines.

\newcommand{\Cmd}{\mathrm{Cmd}}
\newcommand{\Ret}{\mathrm{Ret}}
\newcommand{\StepE}{\operatorname{step}}

\parag{Counter machines}
Let $\Cmd\defeq\{\Op{mul}_2,\Op{inc},\Op{divmod}_2,\Op{iszero}\}$ denote a set of commands, and $\Ret\defeq\{\Op{done}_0,\Op{done}_1\}$ a set of completion statuses. A \emph{multiplicative counter machine with $k$ counters ($k$-CM)} $A=(S,\mathcal{T}_1,\mathcal{T}_2)$ consists of a finite set of states $S$ with $\Op{init},0,1\in S$ and two transition functions $\mathcal{T}_1,\mathcal{T}_2$ mapping a state $q\in S$ to a tuple $(i,j,q'_0,q'_1)$ where $i\in\{1,...,k\}$ refers to a counter, $j\in\Cmd$ is a command, and $q'_0,q'_1\in S$ are successor states ($q'_1$ is not used for $\Op{mul}_2$ and $\Op{inc}$ operations). Additionally, we require that $\mathcal{T}_1,\mathcal{T}_2$ map $q\in\{0,1\}$ to $(1,\Op{iszero},q,q)$, effectively executing no operation from those states.

The idea is that $A$, starting in state $\Op{init}$, picks transitions uniformly at random from either $\mathcal{T}_1$ or $\mathcal{T}_2$. Apart from this randomness, the transitions are deterministic. Eventually, $A$ ends up in either state $0$ or $1$, at which point it cannot perform further actions, thereby indicating whether the input is accepted or rejected.

\parag{Step-execution function}
A \emph{CM-configuration} is a tuple $K=(q,x_1,...,x_k)\in Q\times\mathbb{N}^k$. We define the \emph{step-execution function} $\StepE$ as follows, with $x\in\mathbb{N}$:
\begin{itemize}
\item $\StepE(\Op{mul}_2,x)\defeq (\Op{done}_0,2x)$,
\item $\StepE(\Op{inc},x)\defeq (\Op{done}_0,x+1)$,
\item $\StepE(\Op{divmod},2x+b)\defeq (\Op{done}_b,x)$, for $b\in\{0,1\}$, and
\item $\StepE(\Op{iszero},x)\defeq (\Op{done}_b,x)$, where $b$ is $1$ if $x>0$ and $0$ else.
\end{itemize}
For two CM-configurations $K=(q,x_1,...,x_k)$ and $K'=(q',x_1',...,x_k')$ where $\mathcal{T}_\circ(q)=(i,j,q'_0,q'_1)$ for $\circ\in\{1,2\}$ we write $K\trans{\circ} K'$ if $\StepE(j,x_i)=(\Op{done}_b,x_i')$, $q'=q'_b$ for some $b\in\{0,1\}$, and $x_r=x_r'$ for $r\ne i$. Note that for each $K$ and $\circ$ there is exactly one $K'$ with $K\trans{\circ}K'$.

The reasoning for introducing the step-execution function is that we want to construct a broadcast protocol (BP) which simulates just one step of the CM. Later on we can use this BP as a building block in a more general protocol.

\parag{Computation}
Let $\varphi:\mathbb{N}^l\rightarrow\{0,1\}$ denote a predicate, for $l\le k$, and $C\in\mathbb{N}^l$ an input to $\varphi$. We sample a \emph{random (CM-)execution $\pi=K_0K_1K_2...$ for input $C$}, where $K_0,...$ are CM-configurations, via a Markov chain. For the initial configuration we have $K_0\defeq (\Op{init},C(1),...,C(l),0,...,0)$, and $K_i$ is determined as the unique configuration with $K_{i-1}\trans{\circ}K_i$, where $\circ\in\{1,2\}$ is chosen uniformly at random. (So $\pi$ is the random variable defined as trace of the Markov Chain.)

We say that $A$ \emph{computes $\varphi$ within $f(n)$ steps} if for each $C\in \mathbb{N}^l$ with $|C|=n$ the random execution for input $C$ reaches a configuration in $\{\varphi(C)\}\times\mathbb{N}^k$ after at most $f(n)$ steps in expectation. Finally, we say that $A$ is \emph{$n$-bounded} if the random executions for inputs $C$ with $|C|=n$ can only reach configurations in $Q\times\mathbb{N}^k_{\le n}$.

\renewcommand{\O}{\mathcal{O}}

\begin{theorem}\label{thm:tm-to-cm}
  Let $\varphi$ be a predicate decidable by a log-space bounded randomised Turing machine within $\O(f(n))$ steps in expectation with unary input encoding. There exists an $n$-bounded CM that accepts $\varphi$ within $\O(f(n)\log(n))$ steps in expectation.
\end{theorem}
\begin{proof}[\Proofsketch]
  This can be shown by first representing the Turing machine by a stack machine with two stacks that contain the tape content to the left/right of the current machine head position. In this representation, head movements and tape updates amount to performing pop/push operations on the stack. Moreover, we can simulate an $c\cdot n$-bounded stack by $c$ many $n$-bounded stacks. An $n$-bounded stack, in turn, can be represented in a counter machine with a constant number of $2^n$-bounded counters. The stack content is represented as the base-$2$ number corresponding to the binary sequence stored in the stack. Popping then amounts to a $\Op{divmod}_2$ operation, and pushing amounts to doubling the counter value, followed by adding $1$ or $0$, respectively.

  A detailed proof can be found in Appendix~\ref{app:proof-thm-tm-to-cm}.
\end{proof}

We formally define two types of BPs, ones that simulate a step of the CM, and ones behaving like a clock.

\begin{definition}\label{def:cmsimulation}
Let BP $\PP=(Q\times G,\delta)$ denote a BP with global states $G$ where $0,1,\perp\in Q$ and $\Cmd,\Ret\subseteq G$. We define the injection $\varphi:G\times\mathbb{N}_{\le n}\rightarrow\mathbb{N}^{Q\times G}$ as $\varphi(j,x)\defeq x\cdot\multiset{(1,j)}+(n-x)\cdot\multiset{(0,j)}$. The configurations in $\varphi(\Cmd\times\mathbb{N})$ are called \emph{initial}, the ones in $\varphi(\Ret\times\mathbb{N})$ \emph{final}. We call a configuration $C$ \emph{failing}, if $C(\perp,i)>0$ for some $i\in G$.

We say that $\PP$ is \emph{CM-simulating} if the sets of final and failing configurations are closed under reachability, and from every initial configuration $\varphi(j,w)$ the only reachable final configuration is $\varphi(\StepE(j,w))$, if both are well-defined.
\end{definition}

\newcommand{\Time}{\operatorname{Time}}
\begin{definition}
Let $\PP=(Q,\delta)$ denote a BP with $0,1\in Q$ and $\Time(\PP)$ the number of steps until $\PP$, starting in configuration $\multiset{0,...,0}$, reaches $\multiset{1,...,1}$, or $\infty$ if it does not. If $\Time(\PP)$ is almost surely finite and no agent is in state $1$ before $\Time(\PP)$, then we call $\PP$ a \emph{clock-BP}.
\end{definition}

Now we begin by constructing a CM-simulating BP. The value of a given counter is scattered across the population: each agent stores its contribution to this counter value in its state. The counter value is the sum of all contributions. Usually, an agent's contribution is either $1$ or $0$, thus $n$ agents can maximally store a counter value equal to
$n$, which is not problematic, since the counter machine is assumed to be $n$-bounded.
The difficult part is multiplying and dividing the counter by two. Besides contributions $0$ and $1$, we will also allow intermediate contributions $\frac{1}{2}$  and $2$. By executing a single broadcast, we can multiply (or divide) all the individual contributions by $2$, by setting all contributions of value $1$ to $\frac{1}{2}$, or $2$, respectively. Then, over time, we “normalise” the agents to all have contribution $0$ or $1$ again in a manner which is specified below. This process takes some time, and we cannot determine with perfect reliability whether it is finished, so we only bound the time with high probability. Here and in the following, we say that some event (dependent on the population size $n$) happens \emph{with high probability}, if for \emph{all} $k>0$ the event happens with probability $1-\mathcal{O}(n^{-k})$.

In this and subsequent lemmata we use $\mathcal{G}(p)$, for $0<p<1$, to denote the geometric distribution, that is the number of \emph{trials} until a coin flip with probability $p$ succeeds, which has expectation $1/p$. We start with a statement about the tail distributions of sums of geometric variables.

\begin{restatable}{lemma}{restateharmonictail}\label{lem:harmonictail}
Let $n\ge 3$ and $X_1,...,X_n$ denote independent random variables with sum $X$ and $X_i\sim\mathcal{G}(i/n)$. Then for any $k\ge 1$ there is an $l$ s.t.\
\[\mathbb{P}(X\ge l\cdot n\ln n)\le n^{-k}\]
\end{restatable}
\begin{proof}
See Appendix~\ref{app:harmonictail}.
\end{proof}

\begin{lemma}\label{lem:simulatepoorly}
There is a CM-simulating BP s.t.\ starting from an initial configuration it reaches a final configuration within $\mathcal{O}(n\log n)$ steps with high probability.
\end{lemma}
\begin{proof}
Let $\PP=(Q\times G,\delta)$ denote our BP, with $Q\defeq\{0,\frac{1}{2},1,2,*\}$ and $G\defeq\Cmd\cup\Ret\cup\{\Op{high}\}$. The following transitions initialise the computation, with $b\in\{0,1\}$:
\begin{align}
& (b,\Op{mul}_2)\mapsto(2b,\Op{done}_0),\,\{1\mapsto2,0\mapsto0\} \tag{$\alpha_1$}\\
& (b,\Op{divmod}_2)\mapsto(\tfrac{b}{2},\Op{done}_0),\,\{1\mapsto\tfrac{1}{2},0\mapsto0\} \tag{$\alpha_2$}\\
& (b,\Op{inc})\mapsto(b,\Op{high}),\,\emptyset \tag{$\alpha_3$}
\end{align}
Additionally, we need transitions that move agents back into states $0$ and $1$.
\begin{align}
& (0,\Op{high})\mapsto(1,\Op{done}_0),\,\emptyset \tag{$\beta_1$}\\
& (2,\Op{done}_0)\mapsto(1,\Op{high}),\,\emptyset \tag{$\beta_2$}\\
& (\tfrac{1}{2},\Op{done_0})\mapsto(0,\Op{done}_1),\,\emptyset \tag{$\beta_3$}\\
& (\tfrac{1}{2},\Op{done_1})\mapsto(1,\Op{done}_0),\,\emptyset \tag{$\beta_4$}
\end{align}
This requires some explanation. Basically, we have the invariant that for a configuration $C$ the current value of the counter is $b+\sum_{i\in Q,j\in G}i\cdot C((i,j))$, where $b$ is $1$ if the global state is $\Op{high}$ and $0$ else. There is a “canonical” representation of each counter value, where $b=0$ and the individual contributions $i\in Q$ are only $0$ and $1$. The transitions ($\alpha_1$-$\alpha_3$) update the represented counter value in a single step, but cause a “noncanonical” representation. The transitions ($\beta_1$-$\beta_4$) preserve the value of the counter and cause the representation to eventually become canonical.

This corresponds to final configurations from Definition~\ref{def:cmsimulation}: as long as the representation is noncanonical, i.e.\ an agent with value $\frac{1}{2}$, $2$ or $*$ exists, the configuration is not final. Conversely, once we reach a final configuration our representation is canonical, and, as the value of the counter is preserved, we reach the correct final configuration.
\begin{align}
& (1,\Op{iszero})\mapsto(1,\Op{done}_1),\,\emptyset \tag{$\alpha_4$}\\
& (0,\Op{iszero})\mapsto(0,\Op{done}_0),\,\{1\mapsto *\} \tag{$\alpha_5$}\\
& (*,\Op{done}_0)\mapsto(1,\Op{done}_1),\,\{*\mapsto 1\} \tag{$\beta_5$}
\end{align}
For $\Op{iszero}$ we do something similar, but the value of the counter does not change. If the initial transition is executed by an agent with value $1$, we can go to the global state $\Op{done}_1$ directly. Otherwise, we replace $1$ by $*$ and go to $\Op{done}_0$, so if no agents with value $1$ exist, we are finished. Else some agent with value $*$ executes ($\beta_5$) and we move to the correct final configuration.

Final configurations can only contain states $\{0,1\}\times\Ret$. As we have no outgoing transitions from those states, they are indeed closed under reachability.

It remains to be shown that starting from an initial configuration $C_0$ we reach a final configuration within $\mathcal{O}(n\log n)$ steps with high probability. First, note that transitions ($\alpha_1$-$\alpha_5$) are executed at most once. Moreover, these are the only transitions enabled at $C_0$, so let $C_1$ denote the successor configuration after executing ($\alpha_1$-$\alpha_5$), i.e.\ $C_0\rightarrow C_1$. From now on, we consider only transitions ($\beta_1$-$\beta_5$).

Let $M\defeq\{\tfrac{1}{2},2,*\}\times G$ denote the set of “noncanonical” states, and, for a configuration $C$, let $\Phi(C)\defeq 2\sum_{q\in M}C(q)+b$ denote a potential function, with $b$ being $1$ if the global state of $C$ is $\Op{high}$ and 0 else. Now we can observe that executing a ($\beta_1$-$\beta_5$) transition strictly decreases $\Phi$, and that $0\le\Phi(C)\le 2n$ for any configuration $C$. So after at most $2n$ non-silent transitions, we have reached a final configuration.

Fix some transition ($\beta_j$), let $q\in Q\times G$ denote the state initiating ($\beta_j$), and let $C,C',C''$ denote configurations with $C\trans[3pt]{\beta_j}C'\trans[3pt]{*}C''$, meaning that $C''$ is a configuration reachable from $C$ after executing ($\beta_j$). Then, we claim, $C(q)>C''(q)$.

To see that this holds for transitions ($\beta_2$-$\beta_5$), note that for $i\in\{\frac{1}{2},2,*\}$ the number of agents with value $i$ can only decrease when executing transitions ($\beta_1$-$\beta_5$). For ($\beta_1$) this is slightly more complicated, as ($\beta_3$) increases the number of agents with value $0$. However, ($\beta_1$) is reachable only after ($\alpha_1$) or ($\alpha_3$) has been executed, while ($\beta_3$) requires ($\alpha_2$). Thus, our claim follows.

Let $X_k$ denote the number of silent transitions before executing ($\beta_j$) for the $k$-th time, $k=1,...,l$, and let $r_k$ denote the number of agents in state $q$ at that time. Then $n\ge r_1>r_2>...>r_l\ge 1$ and $X_k$ is distributed according to $\mathcal{G}(r_k/n)$. So we can use Lemma~\ref{lem:harmonictail} to show that the sum of $X_k$ is $\mathcal{O}(n\log n)$ with high probability. There are only $5$ transitions ($\beta_j$), so the same holds for the total number of steps until reaching a final state.
\end{proof}

Our next construction is the clock-BP, which indicates that some amount of time has passed (with high probability). Angluin, Aspnes and Eisenstat used epidemics for this purpose \cite{angluin2008fast}, as do we. The idea is that one agent initiates an epidemic and waits until it sees an infected agent for the first time. Similar to standard analysis of the coupon collector's problem, this is likely to take $\Theta(n\log n)$ time. In the remainder of the proof, we use $\mathbb{P}(x)$ to denote the probability of a stochastic event $x$. By $\mathbb{E}(X)$ we denote the expected value of random variable $X$.

\begin{lemma}\label{lem:simulateclock}
There is a clock-BP $\PP=(Q,\delta)$ s.t.\ $\mathbb{E}(\Time(\PP))\in\mathcal{O}(n\log n)$ and $\Time(\PP)\in\Omega(n\log n)$ with probability $1-\mathcal{O}(n^{-1/2})$.
\end{lemma}
\begin{proof}
For a clock we use states $\{0, 1, c_1, c_2, c_3, c_1^+, c_2^+\}$ and transitions

\begin{align}
& 0\mapsto c_1^+, \{0\mapsto c_2^+\} \tag{$\alpha$}\\
& c_2^+\mapsto c_3, \{c_2^+\mapsto c_2, c_1^+\mapsto c_1\} \tag{$\beta$}\\
& c_3\mapsto c_3, \{c_2\mapsto c_2^+, c_1\mapsto c_1^+\} \tag{$\gamma$}\\
& c_1^+\mapsto 1, \{c_2^+\mapsto 1, c_3\mapsto 1\} \tag{$\omega$}
\end{align}

State $0$ is the initial state, $1$ the final state. States $c_1$ and $c_2$ denote “uninfected” agents, state $c_3$ “infected” ones. The former can become activated (moving to $c_1^+$ and $c_2^+$), causing one of them to become infected. Transition ($\alpha$) marks a leader $c_1$, once they are infected the clock ends (via ($\omega$)). In ($\beta$), a single activated agent becomes infected, deactivating the other agents. They get activated again via transition ($\gamma$). The state diagram is shown in Figure~\ref{fig:infection}.

\begin{figure}[t]
\centering
\def\svgwidth{11cm}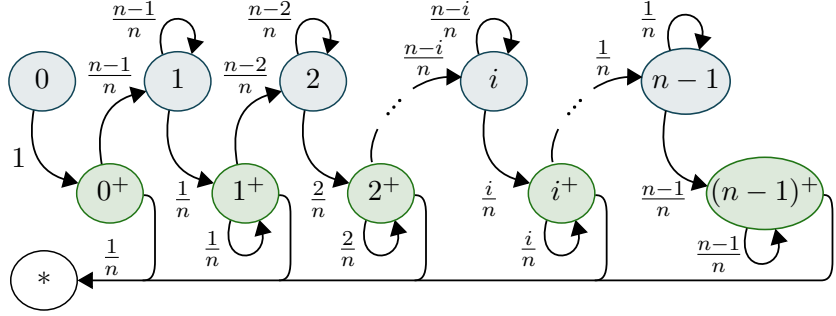
\caption{State diagram of the clock implementation. Nodes with $i$ agents in state $c_3$ are labelled $i$ or $i^+$, the latter denoting that the other agents are in states $c_1^+$ and $c_2^+$. The final state $\ast$ has all agents in state $1$. Arcs are labelled with transition probabilities.}
\label{fig:infection}
\end{figure}

Let $A_i$ denote the event that the leader is the $i$-th infected agent, $i=1,...,n$, i.e.\ the last non-final state we visit in Figure~\ref{fig:infection} is $(i-1)^+$. Due to symmetry, we have $\mathbb{P}(A_i)=1/n$ for all $i$. Let us now condition on $A_i$. The total number of interactions is $X\defeq 2+\sum_{j<i} (r_j+r_j^+)$, where $r_j$ ($r_j^+$) is the number of times the self-loop of state $j$ ($j^+$) is taken, plus one. So $r_j$ and $r_j^+$ are independent random variables, with $r_j\sim\mathcal{G}(j/n)$ and $r_j^+\sim\mathcal{G}(1-j/n)$ (recall that $\mathcal{G}(p)$ is the geometric distribution with mean $1/p$). For the expected value we get $\mathbb{E}(X)=2+\sum_{j<i} (n/j+n/(n-j))$. Using $H_n$ for the $n$-th harmonic number, this is $2+n\cdot(H_{i-1}+H_{n-1}-H_{n-i})\le 4+2n\ln n$. Thus the clock takes $\mathcal{O}(n\log n)$ steps in expectation.

If the leader is one of the first agents infected, the number of steps is too low. For our analysis to give a good bound we will exclude this case, i.e.\ assume that $i>\sqrt{n}$.

Then we have $X\ge r_1+...+r_{\lfloor\sqrt{n}\rfloor}$ and apply Theorem~\ref{thm:jansen31}. Due to $\mathbb{E}(X)\ge n/2\cdot\ln n$
\footnote{This is not quite obvious: $\mathbb{E}(X)/n\ge H_{\sqrt{n}-1}\ge \ln(\sqrt{n}-1)+1/2\ge \ln(\sqrt{n})$, for $n\ge 3$.}
and $1/7-1-\ln 1/7\ge 1$, we get
\[\mathbb{P}(X\le n/14\cdot\ln n)\le e^{-1/n\,\cdot\, n/2\ln n}=n^{-1/2}\]

With probability at most $n^{-1/2}$ we have $i\le\sqrt{n}$, otherwise we get the above bound. So in total $X> n/14\cdot\ln n$ with probability at least $1-2n^{-1/2}$.
\end{proof}

While the above clock measures some interval of time with some reliability, we want a clock that measures an “arbitrarily long” interval with “arbitrarily high” reliability. Constructions for population protocols use phase clocks for this purpose, but broadcasts allow us to synchronise the agents, so we can directly execute the clock multiple times in sequence instead.

\begin{restatable}{lemma}{restatesimulateclockwell}\label{lem:simulateclockwell}
Let $k\in\N$ denote some constant. Then there is a clock-BP $\PP$ s.t.\ $\mathbb{E}(\Time(\PP))\in\mathcal{O}(n\log n)$, and $\Time(\PP)<kn\log n$ with probability $\mathcal{O}(n^{-k})$.
\end{restatable}
\begin{proof}[\Proofsketch]
The idea is that we run $28k^2$ clocks in sequence, in groups of $2k$. Then it is likely that at least one clock in each group works, yielding the overall minimum running time. The full proof can be found in Appendix~\ref{app:simulateclockwell}.
\end{proof}

As mentioned earlier, we combine the clock with the construction in Lemma~\ref{lem:simulatepoorly}. While we cannot reliably determine whether the operation has finished, we can use a clock to measure an interval of time that is sufficiently long for the protocol to terminate with high probability. The next construction does just that. In particular, in contrast to Lemma~\ref{lem:simulatepoorly}, it uses its global state to indicate that it is done.

\begin{lemma}\label{lem:simulatewell}
There is a CM-simulating BP s.t.\ starting from an initial configuration it reaches either a final or a failing configuration $C$ almost surely and within $\mathcal{O}(n\log n)$ steps in expectation, and $C$ is final with high probability. Additionally, all reachable configurations with global state in $\Ret$ are final or failing.
\end{lemma}
\begin{proof}
Fix some $k\in\mathbb{N}$ and let $\PP=(Q\times G,\delta)$ denote the BP we want to construct. Further, let $\PP_1=(Q_1\times G_1,\delta_1)$ denote the BP from Lemma~\ref{lem:simulatepoorly} and choose some $c$ s.t.\ $\PP_1$ reaches a final configuration after at most $cn\log n$ steps with probability at least $1-n^{-k}$.

Now we use Lemma~\ref{lem:simulateclockwell} to get a clock $\PP_2=(Q_2,\delta_2)$ that runs for at least $cn\log n$ steps with probability at least $1-n^{-k}$.

We do a parallel composition of $\PP_1$ and $\PP_2$ to get $\PP$ (as in Appendix~\ref{app:bppforboolean}). In particular, $Q\defeq Q_1\times Q_2$, $G\defeq \{j_\circ:j\in G_1\}\cup\Ret$, where for $Q$ we identify $(i,0)$ with $i$ for $i\in\{0,1\perp\}$, and for $G$ we identify $j$ with $j_\circ$ for $j\in\Cmd$.

Intuitively, we use $\circ$ to rename the global states of $\PP_1$, meaning that the global state $j\in G_1$ of $\PP_1$ is now called $j_\circ$ in our protocol. We want $\PP_1$ to start with the same initial state we have, which is why we identified $j$ with $j_\circ$ for $j\in\Cmd$. However, we only want to enter a final configurations once the clock has run out, so the completion statuses of $\PP_1$ are renamed into $j_\circ$ for $j\in\Ret$ and we enter a final configuration by setting to global state to a $j\in\Ret$.

For each $(q_1,j)\in Q_1\times G_1$ and $q_2\in Q_2$ with $\delta_1(q_1,j)=((r_1,j'),f_1)$ and $\delta_2(q_2)=(r_2,f_2)$ we get the transition
\begin{align}
(q_1,q_2,j_\circ)\mapsto(r_1,r_2,j_\circ'),\,\{(t_1,t_2)\mapsto(f_1(t_1),f(t_2)):t_1\in Q_1, t_2\in Q_2\}\tag{$\alpha$}
\end{align}
These transitions, together with the way we identified states, ensure that $\PP_1$ and $\PP_2$ run normally, with the input being passed through to $\PP_1$ transparently. However, note that the final configurations of $\PP_1$ are not final for $\PP$, meaning that the protocol never ends. Hence, for $q_1\in Q_1,j\in\Ret$ we add the transition
\begin{gather}
\begin{aligned}
(q_1,1,j_\circ)\mapsto(q_1,0,j),\,&\{(b,1)\mapsto(b,0):b\in\{0,1\}\}\\
&\cup\{(i,1)\mapsto(\perp,0):i\in Q_1\setminus\{0,1\}\}
\end{aligned}
\tag{$\beta$}
\end{gather}
This terminates the protocol once the clock has run out. If $\PP_1$ was in a final state, we will now enter a final state as well, else we move into a failing state.
\end{proof}

Finally, we use the above BP to simulate the full $l$-CM.

\begin{restatable}{lemma}{restatecminpolybcp}\label{lem:cminpolybcp}
Fix some predicate $\varphi:\mathbb{N}^k\rightarrow\{0,1\}$ computable by an $n$-bounded $l$-CM within $\mathcal{O}(f(n))\subseteq\mathcal{O}(\operatorname{poly}n)$ steps. Then there is a BCP computing $\varphi$ in $\mathcal{O}(f(n)\,n\log n)$ steps.
\end{restatable}
\begin{proof}[\Proofsketch]
For each counter we need $n$ agents, so $ln$ in total, but we can simply have each agent simulate a constant number of agents. To execute a step of the CM, we use the BP from Lemma~\ref{lem:simulatewell}. It succeeds only with high probability, but in the case of failure at least one agent will have local state $\perp$, from which that agent initiates a restart of the whole computation.

As the CM takes only a polynomial number of steps, we can fix a $k$ s.t.\ a computation of our BCP without failures (i.e.\ one that succeeds on the first try) takes $\mathcal{O}(n^k)$ steps. A single step succeeds with high probability, so we can require it to fail with probability at most $\mathcal{O}(n^{-k-1})$. In total, the restarts increase the running time by a factor of $1/(1-\mathcal{O}(n^{-1}))$, which is only a constant overhead.

The full proof can be found in Appendix~\ref{app:cminpolybcp}.
\end{proof}

This completes the proof of Theorem~\ref{thm:polybppiszl}. By Theorem~\ref{thm:tm-to-cm}, each predicate in $\mathsf{ZPL}$ (with input encoded in unary) is computable by a bounded $l$-CM. Lemma~\ref{lem:cminpolybcp} then yields a polynomial-time BCP for that predicate.

We remark that our reductions also enable us to construct efficient BPPs for specific predicates. The predicate \textsc{PowerOfTwo} for example, as described in \cite[Proposition~3]{blondin2019expressive}, can trivially be decided by an $\O(\log n)$-time bounded RTM with input encoded as binary, so there is also a BCP computing that predicate within $\O(n\log^2n)$ interactions.

%% file: figures/infection.pdf_tex
\begingroup%
  \makeatletter%
  \providecommand\color[2][]{%
    \errmessage{(Inkscape) Color is used for the text in Inkscape, but the package 'color.sty' is not loaded}%
    \renewcommand\color[2][]{}%
  }%
  \providecommand\transparent[1]{%
    \errmessage{(Inkscape) Transparency is used (non-zero) for the text in Inkscape, but the package 'transparent.sty' is not loaded}%
    \renewcommand\transparent[1]{}%
  }%
  \providecommand\rotatebox[2]{#2}%
  \newcommand*\fsize{\dimexpr\f@size pt\relax}%
  \newcommand*\lineheight[1]{\fontsize{\fsize}{#1\fsize}\selectfont}%
  \ifx\svgwidth\undefined%
    \setlength{\unitlength}{524.66749777bp}%
    \ifx\svgscale\undefined%
      \relax%
    \else%
      \setlength{\unitlength}{\unitlength * \real{\svgscale}}%
    \fi%
  \else%
    \setlength{\unitlength}{\svgwidth}%
  \fi%
  \global\let\svgwidth\undefined%
  \global\let\svgscale\undefined%
  \makeatother%
  \begin{picture}(1,0.39052327)%
    \lineheight{1}%
    \setlength\tabcolsep{0pt}%
    \put(0,0){\includegraphics[width=\unitlength,page=1]{infection.pdf}}%
    \put(0.04448559,0.27961794){\color[rgb]{0,0,0}\makebox(0,0)[t]{\lineheight{1.25}\smash{\begin{tabular}[t]{c}$0$\end{tabular}}}}%
    \put(0,0){\includegraphics[width=\unitlength,page=2]{infection.pdf}}%
    \put(0.206568,0.27961794){\color[rgb]{0,0,0}\makebox(0,0)[t]{\lineheight{1.25}\smash{\begin{tabular}[t]{c}$1$\end{tabular}}}}%
    \put(0,0){\includegraphics[width=\unitlength,page=3]{infection.pdf}}%
    \put(0.12985156,0.14454923){\color[rgb]{0,0,0}\makebox(0,0)[t]{\lineheight{1.25}\smash{\begin{tabular}[t]{c}$0^+$\end{tabular}}}}%
    \put(0,0){\includegraphics[width=\unitlength,page=4]{infection.pdf}}%
    \put(0.36865009,0.27961794){\color[rgb]{0,0,0}\makebox(0,0)[t]{\lineheight{1.25}\smash{\begin{tabular}[t]{c}$2$\end{tabular}}}}%
    \put(0,0){\includegraphics[width=\unitlength,page=5]{infection.pdf}}%
    \put(0.2919337,0.14454923){\color[rgb]{0,0,0}\makebox(0,0)[t]{\lineheight{1.25}\smash{\begin{tabular}[t]{c}$1^+$\end{tabular}}}}%
    \put(0,0){\includegraphics[width=\unitlength,page=6]{infection.pdf}}%
    \put(0.04664893,0.0418829){\color[rgb]{0,0,0}\makebox(0,0)[t]{\lineheight{1.25}\smash{\begin{tabular}[t]{c}$\ast$\end{tabular}}}}%
    \put(0,0){\includegraphics[width=\unitlength,page=7]{infection.pdf}}%
    \put(0.45401433,0.14454923){\color[rgb]{0,0,0}\makebox(0,0)[t]{\lineheight{1.25}\smash{\begin{tabular}[t]{c}$2^+$\end{tabular}}}}%
    \put(0,0){\includegraphics[width=\unitlength,page=8]{infection.pdf}}%
    \put(0.58476001,0.27961794){\color[rgb]{0,0,0}\makebox(0,0)[t]{\lineheight{1.25}\smash{\begin{tabular}[t]{c}$i$\end{tabular}}}}%
    \put(0,0){\includegraphics[width=\unitlength,page=9]{infection.pdf}}%
    \put(0.67012429,0.14454923){\color[rgb]{0,0,0}\makebox(0,0)[t]{\lineheight{1.25}\smash{\begin{tabular}[t]{c}$i^+$\end{tabular}}}}%
    \put(0,0){\includegraphics[width=\unitlength,page=10]{infection.pdf}}%
    \put(0.81437893,0.27961794){\color[rgb]{0,0,0}\makebox(0,0)[t]{\lineheight{1.25}\smash{\begin{tabular}[t]{c}$n-1$\end{tabular}}}}%
    \put(0,0){\includegraphics[width=\unitlength,page=11]{infection.pdf}}%
    \put(0.90892527,0.14454923){\color[rgb]{0,0,0}\makebox(0,0)[t]{\lineheight{1.25}\smash{\begin{tabular}[t]{c}$(n-1)^+$\end{tabular}}}}%
    \put(0,0){\includegraphics[width=\unitlength,page=12]{infection.pdf}}%
    \put(0.01662139,0.18459564){\color[rgb]{0,0,0}\makebox(0,0)[t]{\lineheight{1.25}\smash{\begin{tabular}[t]{c}$1$\end{tabular}}}}%
    \put(0.127842,0.06667069){\color[rgb]{0,0,0}\makebox(0,0)[t]{\lineheight{1.25}\smash{\begin{tabular}[t]{c}$\frac{1}{n}$\end{tabular}}}}%
    \put(0,0){\includegraphics[width=\unitlength,page=13]{infection.pdf}}%
    \put(0.24812027,0.07666711){\color[rgb]{0,0,0}\makebox(0,0)[t]{\lineheight{1.25}\smash{\begin{tabular}[t]{c}$\frac{1}{n}$\end{tabular}}}}%
    \put(0,0){\includegraphics[width=\unitlength,page=14]{infection.pdf}}%
    \put(0.12539973,0.28595201){\color[rgb]{0,0,0}\makebox(0,0)[t]{\lineheight{1.25}\smash{\begin{tabular}[t]{c}$\frac{n-1}{n}$\end{tabular}}}}%
    \put(0,0){\includegraphics[width=\unitlength,page=15]{infection.pdf}}%
    \put(0.21287404,0.13628097){\color[rgb]{0,0,0}\makebox(0,0)[t]{\lineheight{1.25}\smash{\begin{tabular}[t]{c}$\frac{1}{n}$\end{tabular}}}}%
    \put(0,0){\includegraphics[width=\unitlength,page=16]{infection.pdf}}%
    \put(0.28748125,0.28595201){\color[rgb]{0,0,0}\makebox(0,0)[t]{\lineheight{1.25}\smash{\begin{tabular}[t]{c}$\frac{n-2}{n}$\end{tabular}}}}%
    \put(0,0){\includegraphics[width=\unitlength,page=17]{infection.pdf}}%
    \put(0.50077384,0.30722426){\color[rgb]{0,0,0}\makebox(0,0)[t]{\lineheight{1.25}\smash{\begin{tabular}[t]{c}$\frac{n-i}{n}$\end{tabular}}}}%
    \put(0,0){\includegraphics[width=\unitlength,page=18]{infection.pdf}}%
    \put(0.37495471,0.13628097){\color[rgb]{0,0,0}\makebox(0,0)[t]{\lineheight{1.25}\smash{\begin{tabular}[t]{c}$\frac{2}{n}$\end{tabular}}}}%
    \put(0,0){\includegraphics[width=\unitlength,page=19]{infection.pdf}}%
    \put(0.41020268,0.07666711){\color[rgb]{0,0,0}\makebox(0,0)[t]{\lineheight{1.25}\smash{\begin{tabular}[t]{c}$\frac{2}{n}$\end{tabular}}}}%
    \put(0,0){\includegraphics[width=\unitlength,page=20]{infection.pdf}}%
    \put(0.62631255,0.07666711){\color[rgb]{0,0,0}\makebox(0,0)[t]{\lineheight{1.25}\smash{\begin{tabular}[t]{c}$\frac{i}{n}$\end{tabular}}}}%
    \put(0,0){\includegraphics[width=\unitlength,page=21]{infection.pdf}}%
    \put(0.15459246,0.35352168){\color[rgb]{0,0,0}\makebox(0,0)[t]{\lineheight{1.25}\smash{\begin{tabular}[t]{c}$\frac{n-1}{n}$\end{tabular}}}}%
    \put(0,0){\includegraphics[width=\unitlength,page=22]{infection.pdf}}%
    \put(0.31667486,0.35352168){\color[rgb]{0,0,0}\makebox(0,0)[t]{\lineheight{1.25}\smash{\begin{tabular}[t]{c}$\frac{n-2}{n}$\end{tabular}}}}%
    \put(0,0){\includegraphics[width=\unitlength,page=23]{infection.pdf}}%
    \put(0.53278543,0.35352168){\color[rgb]{0,0,0}\makebox(0,0)[t]{\lineheight{1.25}\smash{\begin{tabular}[t]{c}$\frac{n-i}{n}$\end{tabular}}}}%
    \put(0,0){\includegraphics[width=\unitlength,page=24]{infection.pdf}}%
    \put(0.76993528,0.35352168){\color[rgb]{0,0,0}\makebox(0,0)[t]{\lineheight{1.25}\smash{\begin{tabular}[t]{c}$\frac{1}{n}$\end{tabular}}}}%
    \put(0,0){\includegraphics[width=\unitlength,page=25]{infection.pdf}}%
    \put(0.57899509,0.13720872){\color[rgb]{0,0,0}\makebox(0,0)[t]{\lineheight{1.25}\smash{\begin{tabular}[t]{c}$\frac{i}{n}$\end{tabular}}}}%
    \put(0,0){\includegraphics[width=\unitlength,page=26]{infection.pdf}}%
    \put(0.78744585,0.13720872){\color[rgb]{0,0,0}\makebox(0,0)[t]{\lineheight{1.25}\smash{\begin{tabular}[t]{c}$\frac{n-1}{n}$\end{tabular}}}}%
    \put(0,0){\includegraphics[width=\unitlength,page=27]{infection.pdf}}%
    \put(0.71561914,0.30722426){\color[rgb]{0,0,0}\makebox(0,0)[t]{\lineheight{1.25}\smash{\begin{tabular}[t]{c}$\frac{1}{n}$\end{tabular}}}}%
    \put(0,0){\includegraphics[width=\unitlength,page=28]{infection.pdf}}%
    \put(0.85382632,0.07031851){\color[rgb]{0,0,0}\makebox(0,0)[t]{\lineheight{1.25}\smash{\begin{tabular}[t]{c}$\frac{n-1}{n}$\end{tabular}}}}%
    \put(0,0){\includegraphics[width=\unitlength,page=29]{infection.pdf}}%
  \end{picture}%
\endgroup%

%% file: appendix.tex

\section{Proof of Theorem~\ref{thm:tm-to-cm}: From TMs to Counter Machines}
\label{app:proof-thm-tm-to-cm}
In this section, we sketch the proof that every log-space-bounded probabilistic Turing machine with expected polynomial run-time can be simulated by an  input-bounded, randomised, multiplicative counter machines with negligible runtime overhead.
The proof is a reduction in several steps:
\begin{enumerate}
  \item We start with a randomised Turing machine $M$ that takes input in \emph{unary} encoding and requires \emph{logarithmic} space.
  \item We then show how to transform $M$ into an equivalent RTM $M'$ that takes input in \emph{binary} encoding and requires \emph{linear} space on a single tape.
  \item Then we transform $M'$ into an equivalent multi-stack machine $S$, and $S$ into an equivalent stack machine $S'$ where the length of each stack can be bounded \emph{precisely} by the length of the input.
  \item Finally, we show how to transform $S'$ into an equivalent counter machine $K$.
\end{enumerate}
The rest of this section is structured as follows.
In Subsection~\ref{ssection:turing-zl}, we formally introduce randomised Turing machines. In Subsection~\ref{ssection:stack}, we introduce randomised stack machines and sketch how to simulate randomised Turing machines efficiently in stack machines whose stacks are bounded by the input size. Finally, in Subsection~\ref{ssection:stack-to-counter}, we show how to simulate stack machines in multiplicative counter machines.

\subsection{Turing machines}
\label{ssection:turing-zl}

For our reduction, we need to define a Turing machine model.

\begin{definition}[Randomised Turing Machine]
A randomised Turing machine (RTM) is a tuple $(Q, \Sigma, \Gamma, q_0, q_a, q_r \delta_1, \delta_2)$ where
\begin{itemize}
  \item $Q$ is a finite set of states,
  \item $\Gamma$ is the tape alphabet, with $\square,0,1\in\Gamma$,
  \item $q_0,q_f,q_r \in Q$ are the initial, accepting, and rejecting states, respectively.
  \item $\delta_i \colon Q \times \Gamma \mapsto Q \times \Gamma \times \{-1, 0, +1\}$ are the transition functions.
\end{itemize}
Additionally, we require that $\delta_i(q,\alpha)=(q,\alpha,0)$ for all $q\in\{q_f,q_r\}$, $i\in\{1,2\}$ and $\alpha\in\Gamma$.
\end{definition}
So once the RTM has reached $q_f$ or $q_r$, it performs no further actions.

\parag{Configurations and step relation}
A \emph{configuration} of an RTM $M$ is a tuple $C = (q, i, \tau)$ consisting of a control state $q\in Q$, read/write head position $i \in \Z$, and tape $\tau \colon \Z \rightarrow \Gamma$. The configuration $C$ is acceping if $q = q_f$, and rejecting if $q = q_r$.
For $\circ \in \{1, 2\}$ and two configurations $C = (q, i, \tau)$,  $C'= (q', i', \tau')$, a \emph{step} $C \trans{\circ} C'$ is valid if and only if the following holds:
\begin{itemize}
\item $\delta_{\circ}(q, \tau(i)) = (q', \alpha, d)$ for some $\alpha$ and some $d$,
\item $i' = i + d$,
\item $\tau'(i) = \alpha$ and $\tau'(j) = \tau(j)$ for all $j \neq i$.
\end{itemize}

\parag{Input encoding}
For a given input $x\in\N^k$ and configuration $C\defeq(q_0, 0, \tau)$, we say that $C$ is \emph{$x$ encoded as unary (binary)} if $\tau(1)...\tau(n)=\square1^{x_1}\square...\square1^{x_k}\square$ (resp.\ $\tau(1)...\tau(n)=\square(x_1)_2\square...\square(x_1)_2\square$), where $(r)_2$ denotes the binary representation of $r\in\N$, and $\tau(i)=0$ for $i\notin\{1,...,n\}$.

\parag{Random executions}
We define the \emph{random exection for input $x$ encoded as unary (binary)} as a Markov chain $\pi=C_0C_1C_2...$ where $C_0,...$ are RTM-configurations, $C_0$ is $x$ encoded as unary (binary), and $C_i$ is determined as the unique configuration with $C_{i-1}\trans{\circ}C_i$, where $\circ\in\{1,2\}$ is chosen uniformly at random.

\parag{Acceptance/rejection}
For a predicate $\varphi:\N^k\rightarrow\{0,1\}$, we say that $M$ \emph{computes $\varphi$ within $\O(f(n))$ steps}, if for all $x \in \N^k$ the random execution $C_0...$ for input $x$ reaches an accepting or rejecting configuration $C$ almost surely and after at most $f(n)$ steps in expectation, and $C$ is accepting iff $\varphi(x)=1$. Such an RTM is called \emph{$f(n)$-time bounded}. We say $M$ is \emph{$f(n)$-space bounded} if all configurations $(q, i, \tau)$ reachable from $C_0$ satisfy $|i| \leq f(|x|)$.

\parag{Log-space Turing machines}
The definition of the Turing machine model given above does not have a separate reading/writing tape, and is thus not suitable for the definition of log-space complexity classes. For log-space Turing machines we assume a modified model, where instead of one tape $\tau$, we have two tapes: An immutable reading tape $\tau_\text{read}$, and a working tape $\tau_\text{work}$, and instead of one head we have two heads, one for each tape, which can be moved independently. A configuration then is a tuple $(q, i, j, \tau_\text{read}, \tau_\text{work})$, where $q$ is the current state, $i$ is the position of the reading head, and $j$ is the position of the working head. The other definitions are adapted in the obvious manner.

\begin{lemma}\label{lemma:unary-to-binary-tm}
Every predicate decidable by an $\O(\log n)$-space bounded, $\O(f(n))$-time bounded RTM with input encoded as unary is decidable by an $\O(N)$-space bounded and  $\O(N\cdot f(2^N))$-time bounded RTM with input encoded as binary.
\end{lemma}
\begin{proof}[\Proofsketch]
Let $M$ be an $\O(\log n)$-space bounded, $\O(f(n))$-time bounded RTM with input encoded as unary. We sketch the construction of an RTM $M'$ with binary input encoding that simulates $M$ in $\O(N)$ space and within $\O(f(2^N))$ steps in expectation, where $N=\log_2n$ denotes the size of the input of $M'$ for clarity.

Movements and updates on the working tape of $M$ can clearly be simulated by $M'$ within the given space. Further, $M'$ can simulate the movement of the head on the input tape of $M$ by keeping a single binary counter to encode its position and performing a constant number of additions/comparisons when $M$ moves the input head. This means that $M'$ simulates a step of $M$ with an overhead of $N$ steps.

As $M$ takes $f(n)=f(2^N)$ steps, we end up with $M'$ deciding the predicate in $Nf(2^N)$ steps.
\end{proof}

\subsection{Stack Machines}
We now sketch how a RTMs $M$ can be simulated efficiently with randomised stack machines.

\label{ssection:stack}
\begin{definition}
A (randomised) $l$-stack-machine $M$ is a tuple $(Q, \delta_1, \delta_2, q_0, q_f, q_r)$ where
\begin{itemize}
  \item $Q$ is a finite set of states,
  \item $\delta_i \colon Q \rightarrow \left([l] \times \{0,1\} \times Q \right) \cup \left([l] \times Q \times Q \times Q) \right)$
  are transition functions
  \item $q_0,q_f,q_r \in Q$ are the initial, accepting, and rejecting states, respectively.
\end{itemize}
Additionally, we require that $\delta_i(q)=(1,q,q,q)$ for $i\in\{1,2\}$ and $q\in\{q_f,q_r\}$, meaning that from states $q_f,q_r$ no (significant) actions are performed.
\end{definition}
\newcommand{\oppush}[2]{\trans{\textsc{push } {#1}\, {#2}}}
\newcommand{\oppop}[1]{\trans{\textsc{pop } {#1}}}

\parag{Configurations and transitions}
A configuration of the $l$-stack machine $M$ is a tuple $C = (q, s_1, \ldots, s_l)$, consisting of a control state $q \in Q$, and $l$ binary stacks $s_1, \ldots, s_k \in \{0,1\}^*$. We say $C$ is \emph{accepting} if $q = q_f$, and \emph{rejecting} if $q = q_r$. Let $C' = (q', s_1', \ldots, s_l')$ be a configuration of $M$. We use $C \oppush{k}{\alpha}_\circ C'$ and $C \oppop{k}_\circ C'$ to denote that $C'$ results from pushing $\alpha$ to or popping from the $k$th stack, respectively, as specified in $\delta_{\circ}$:
Formally, we write $C \oppush{k}{\alpha}_\circ C'$ whenever $\delta_\circ(q) = (k, \alpha, q')$ and $s_k' = \alpha \cdot s_k$. We write $C \oppop{k}_\circ C'$ whenever $\delta(q) = (k, q_1, q_2, q_3)$ and $s_i = s_i'$ for every $i \neq k$, and one of following three constraints is satisfied:
\begin{align}
  q' & = q_1 \text{ and } s_k = 0 \cdot s_k', \\
  q' & = q_2 \text{ and } s_k = 1 \cdot s_k', \\
  q' & = q_3 \text{ and } s_k = s_k' = \epsilon.
\end{align}
 We write $C \trans{\circ} C'$ to indicate that $C \oppush{k}{\alpha}_\circ C'$ or $C \oppop{k}_\circ C'$ holds.

\parag{Input encoding}
These encodings are technical artefacts of the reductions we use, hence they are quite unintuitive. For a given input $x\in\N^k$ and configuration $C\defeq(q_0, s_1,...,s_l)$, we say that $C$ is \emph{$x$ two-symbol encoded} if $s_1=f(\tau(0))...f(\tau(n))$ and $s_2=...=s_l=\epsilon$, where $\tau$ are the tape contents of $x$ encoded as binary and $f(0)\defeq01$, $f(1)\defeq10$, $f(\square)\defeq11$. We say that $C$ is \emph{$x$ multi-stack encoded} if $s_{ji}=r_{3(i-1)+j}$ for $j\in\{1,2,3\}$ and $i\in\N$ with $1\le3(i-1)+j\le n$, where $(q_0,r,\epsilon,...,\epsilon)$ is $x$ two-symbol encoded.

The idea of the multi-stack encoding is that each stack has length at most $n$, which we can ensure by distributing the symbols in a round-robin fashion.

\parag{Acceptance condition and boundedness}
Random executions and acceptance are defined as in RTMs, now for configurations of the stack and the step relation of the stack machine. An $l$-stack machine $M$ is \emph{$f(n)$-bounded} if any initial configuration given by an encoded input $x$ can only reach configurations $(q,s_1,...,s_l)$ with $|s_1|,...,|s_l|\le f(n)$, where $n\defeq \log_2(x(1)+...+x(k))$. (So $n$ is roughly the length of $x$ in binary encoding.)

\smallskip

For subsequent propositions, we fix a predicate $\varphi:\N^k\rightarrow\{0,1\}$ and a function $f \colon \N \rightarrow \N$. We will now show that every $\O(n)$-bounded stack machine can be simulated by an $n$-bounded stack machine.

\begin{proposition}\label{prop:linear-to-n}
Let $c \in \N$.  For every $cn$-bounded $l$-stack machine that decides $\varphi$ in $\O(f(n))$ steps with input two-symbol encoded, there  exists an $n$-bounded $(c(l-1) + 3)$-stack machine that decides $\varphi$ in $\O(f(n))$ steps with input multi-stack encoded.
\end{proposition}
\begin{proof}[\Proofsketch]
The idea is to replace each stack $s$ of the $c\cdot n$-bounded machine by $c$ many $n$-bounded auxiliary stacks $s_1, ..., s_c$. For each stack $s$ of the original machine, the new machine stores in its control state a round counter $i \in \{1, ..., c\}$ that specifies the currently used auxiliary stack $s_i$, starting with $i=1$. Pushing some symbol $\alpha$ onto $s$ and transitioning from states $q$ to $q'$ is implemented by updating $q$ to $q'$ and updating $i$ to its successor $i'$, where $i' = i + 1$, if $i < c$, else $i' = 1$, and by pushing $\alpha$ onto $s_{i'}$. This way, the first element of $s$ is pushed onto $s_1$, the second onto $s_2$, and so forth, until the element $c+1$, which again goes onto $s_1$, and so forth. Symmetrically, popping from $s$  amounts to popping from $s_i$  and decrementing $i$ (resp.\ setting $i = c$, if $i = 1$).

For the input stack we use three auxiliary stacks instead of $c$, as that suffices for those stacks to be $n$-bounded. Conveniently, the multi-stack encoding means that we do not have to perform any work to convert the input.
\end{proof}

\begin{proposition} \label{prop:tm-to-stack}
Let $M$ be an $\O(n)$-space bounded, $\mathcal{O}(f(n))$-time bounded RTM deciding $\varphi$ with input encoded as binary. Then there exists an $n$-bounded stack machine $S$ that decides $\varphi$ in $\mathcal{O}(f(n))$ steps with input multi-stack encoded.
\end{proposition}
\begin{proof}[\Proofsketch]
We assume wlog that $M$ uses $\Gamma = \{\square,0,1\}$ as tape alphabet. Let $C = (q, i, \tau)$ be a configuration of $M$.
Then $C$ can be represented as triple $(q, l, r)$ where $l,r \in \{0, 1\}^\omega$,
$l \defeq f(\tau(i-1)) f(\tau(i-2))  ...$ represents the  side of the tape left to the head, and $r \defeq f(\tau(i)) f(\tau(i+1)) ...$ represents the side of the tape to the right of the head, where $f$ is defined as for the two-symbol encoding.
Moving the head to the left and updating the tape then amounts to removing the first two left-most symbols from $l$, and prepending the resulting updated digits to $r$, and symmetrically for moving the head to the right.

Since $M$ is $\O(n)$-space bounded, we may represent $l$ and $r$ by linearly bounded stacks in a stack machine $S$. Again, the two-symbol input encoding means that we do not have to perform any work for initialisation, as long as we use $r$ as input stack.

The aforementioned updates can then be performed using corresponding pop/push operations. This gives an $\O(n)$-bounded stack machine $S$ that decides $\varphi$ in $\O(f(n))$ steps. By Proposition~\ref{prop:linear-to-n}, we can transform $S$ into a $n$-bounded stack machine that decides $\varphi$ in $\O(f(n) + n)$ steps.
\end{proof}

\subsection{From Stack Machines to Counter Machines}
\label{ssection:stack-to-counter}
\begin{proposition}
Let $M$ be an $\O(n)$-space bounded RTM that decides $\varphi$ in $\O(f(n))$ steps, where the input is given in unary. Then there exists an $n$-bounded counter machine that computes $\varphi$ in $\O(f(n)\log(n))$ steps.
\end{proposition}
\begin{proof}
By Lemma~\ref{lemma:unary-to-binary-tm} and Proposition~\ref{prop:tm-to-stack}, there exists an $N$-bounded $l$-stack-machine $S$ that decides $\varphi$ in $\O(f(2^N)\cdot N)$ steps with input in multi-stack encoding, where $N=\log_2n$ is roughly the length of the input encoded as binary.

We now describe how the stack content is represented in the counters, how the counters can be initialized, and how push/pop operations can be implemented efficiently in that representation.

\parag{Representing stack content}
Each stack $s_i$ of $S$ is represented by two counters $x_i,x_i'$ in the counter machine, where $x_i$ contains the actual symbols of $s_i$, and $x_i'$ simply contains a 1 for each symbol in $s_i$ and serves to determine whether $x_i$ is empty. Formally, for $s_i=w_0...w_m$ we get $x_i=\sum_j w_j2^j$ and $x_i'=\sum_j 2^j$. As $S$ is $N$-bounded, each stack has at most $N=\log_2n$ symbols, and thus each of our counters is at most $n$ large. Therefore we are indeed $n$-bounded.

Hence we can simply push a symbol $b\in\{0,1\}$ onto $s_i$ by performing operations $\Op{mul}_2$ on $x_i$ and $x_i'$, $\Op{inc}$ on $x_i'$, and if $b=1$ also $\Op{inc}$ on $x_i$. Similarly, to pop from a stack we determine whether it is empty via $\Op{iszero}$ on $x_i'$, and do a $\Op{divmod_2}$ on both $x_i$ and $x_i'$ if it is not.

This means that we simulate a single operation with only constant overhead, so we execute $\O(f(n)\log(n)$ steps of the stack machine.

\parag{Initialization of stacks}
Recall that $S$ accepts input in the multi-stack encoding, while the counter machine simply has $k$ counter initialised to the input values.

We can determine bits of the binary representation of the value stored in a counter by repeatedly performing $\Op{divmod}_2$, which are straightforward to encode into their two-symbol encoding, which we then push onto the correct stacks in round-robin fashion to construct the multi-stack encoding. Here, we use $\Op{iszero}$ to determine whether any bits are left in a counter. This initialisation procedure takes $\O(\log n)$ steps, which is subsumed in the overall running time.
\end{proof}

\section{Proof: Stability can be decided in logarithmic time} \label{app:proof-stable}
Call a configuration $C$ \emph{stable} if it is in a $b$-consensus for some $b \in \{0, 1\}$, and every configuration reachable from $C$ is also in a $b$-consensus.
\begin{proposition}
  For every BCP $\PP=(Q, \Sigma, \delta, I, O)$, there exists a (deterministic) Turing machine $T_\PP$ that decides in constant time whether a given configuration $C \in \N^Q$ is stable.
\end{proposition}
\begin{proof}
  Let $U \subseteq \N^Q$ be the set of unstable configurations of $\PP$. It is straightforward to see that $U$ is upwards-closed, that is, $C \in U$ implies $C' \in U$ for every $C' \geq C$ (if a configuration is not stable, then adding additional agents cannot make it stable).

  By Dickson's lemma~\cite{kruskal1972theory}, $U$ has finitely many minimal elements. Since $\PP$ is not part of the input, we may assume that the minimal elements are precomputed and $T_\PP$ can iterate over the minimal elements in constant time. In order to decide whether $C \in U$ holds, $T_\PP$ only needs to verify whether $C' \leq C$ holds for some minimal unstable configuration $C'$. By the previous considerations, this can be done with a constant number of comparisons, so in time $\mathcal{O}(\log n)$.
\end{proof}

\section{Proof of Proposition~\ref{prop:bppforboolean}: Boolean Combinations of Predicates}\label{app:bppforboolean}
\restatebppforboolean*

\begin{proof}
It suffices to show the statement for $k=2$ and $\varphi=\neg\varphi_1\wedge\neg\varphi_2$.

Let $\PP_i=:(Q_i,\Sigma, \delta_i,I_i,O_i)$, $i=1,2$. We assume that both protocols work on the same input alphabet $\Sigma$. This is without loss of generality, as we can always unify alphabets by adding missing symbols and identity mappings. We define our resulting protocol $\PP=(Q,\Sigma, \delta, I,O)$ as follows:
\begin{align*}
  Q & \defeq Q_1 \times Q_2, \\
  \delta((q_1, q_2)) & \defeq ((r_1, r_2), f) \text{ with } \delta_i(q_i) = (r_i, f_i) \text{ and } f((q_1, q_2)) = (f_1(q_1), f_2(q_2)) \\
  I(x) & \defeq (I(x), I(x)) \text{ for every } x \in \Sigma, \\
  O & \defeq \{(q_1,q_2):q_1\notin O_1\wedge q_2\notin O_2\}
\end{align*}

Every execution of $\PP$ maps onto one of $\PP_1$ and one of $\PP_2$ by projecting onto either the first of the second element of each state, so every fair execution of $\PP$ consists of two executions of $\PP_1$ and $\PP_2$ which stabilize to the value of $\varphi_1$ and $\varphi_2$ within $\mathcal{O}(n\log n)$ steps in expectation, respectively. Hence $\PP$  stabilizes to $\neg\varphi_1(C_0) \land \neg\varphi_2(C_0)=\varphi(C_0)$ within $\mathcal{O}(n\log n)$ steps in expectation.
\end{proof}

\section{Sums of Geometric Random Variables}\label{app:harmonictail}
We use the following theorems to estimate lower and upper tail probabilities for sums of geometrically distributed random variables. As before, we use $\mathcal{G}(p)$, for $0<p<1$, to denote the geometric distribution (with expectation $1/p$).

\begin{theorem}[{see~\cite[Theorem~2.1]{janson2018tail}}] \label{thm:jansen21}
Let $X_1,...,X_n$ denote independent random variables with $X_i\sim\mathcal{G}(p_i)$ for $i=1,...,n$ and $0<p_i\le1$, and set $X\defeq X_1+...+X_n$, $\mu\defeq\mathbb{E}(X)$, and $p_*\defeq\min_ip_i$. Then, for any $\lambda\ge1$
\[\mathbb{P}(X\ge\lambda\mu)\le e^{-p_*\mu(\lambda-1-\ln\lambda)}\]
\end{theorem}
\begin{theorem}[{see~\cite[Theorem~3.1]{janson2018tail}}] \label{thm:jansen31}
Let $X_1,...,X_n$ denote independent random variables with $X_i\sim\mathcal{G}(p_i)$ for $i=1,...,n$ and $0<p_i\le1$, and set $X\defeq X_1+...+X_n$, $\mu\defeq\mathbb{E}(X)$, and $p_*\defeq\min_ip_i$. Then, for any $\lambda\le1$
\[\mathbb{P}(X\le\lambda\mu)\le e^{-p_*\mu(\lambda-1-\ln\lambda)}\]
\end{theorem}

In the following lemma we use the first of these bounds to show that the sum $X_1+...+X_n$ with $X_i\sim\mathcal{G}(i/n)$ is in $\Omega(n\log n)$ with high probability. This is the same analysis as used for the coupon collector's problem.

\restateharmonictail*
\begin{proof}
First, note that $\mathbb{E}(X)=nH_n$, where $H_n:=1/1+...+1/n$ is the $n$-th harmonic number. Using Theorem~\ref{thm:jansen21} with $\lambda$ chosen s.t.\ $\lambda-1-\ln\lambda=k$ we get
\[\mathbb{P}(X\ge \lambda nH_n)\le e^{-1/n\,\cdot\,nH_nk}=e^{-kH_n}\]
We have $\ln n\le H_n\le 1+\ln n \le 2\ln n$, for $n\ge 3$, which yields the desired bound when choosing $l:=2\lambda$.
\end{proof}

\section{Proof of Lemma~\ref{lem:simulateclockwell}: Arbitrarily Good Clocks}\label{app:simulateclockwell}
\restatesimulateclockwell*
\begin{proof}
To simplify notation, we will use the precise constants proved in Lemma \ref{lem:simulateclock} here, although the proof does not require them.

Choose $l\defeq 28k^2$ as the number of clocks we want to run in sequence and let $(Q',\delta')$ denote the BP from Lemma~\ref{lem:simulateclock}. We construct our BP $\PP$ as $(Q'\times\{1,...,l\},\delta)$, with global states $\{1,...,l\}$ and transitions $\delta$ as follows.
\begin{align}
& (q,i)\mapsto(r,i),\,f\qquad&&\text{for $\delta(q)=(r,f),i=1,...,l$}\tag{$\alpha$} \\
& (1,i)\mapsto(0,i+1),\{1\mapsto0\}&&\text{for $i=1,...,l-1$} \tag{$\beta$}
\end{align}
Additionally, we identify state $(0,0)$ with $0$ and $(1,l)$ with $1$.

For the analysis, we are going to divide the $l$ clocks into groups of $2k$. As shown in Lemma \ref{lem:simulateclock}, each clock has at most a $2n^{-1/2}$ probability of “failing”, i.e.\ taking fewer than $n/14\cdot\ln n$ steps. So the probability that all $2k$ clocks in a group fail is at most $2^{2k}n^{-k}$.

By union bound, the probability that there exists a group which has all of its clocks failing is at most $14\cdot 2^{2k}kn^{-k}$. Conversely, if each group has a single non-failing clock, we take at least $kn\ln n$ steps in total.
\end{proof}

\section{Proof of Lemma~\ref{lem:cminpolybcp}: Simulating CMs efficiently}\label{app:cminpolybcp}
\restatecminpolybcp*
\begin{proof}
Let $A=(S,\mathcal{T}_1,\mathcal{T}_2)$ denote the CM. We will construct our BCP assuming to work with $ln$ agents instead of $n$, as each agent can simulate a constant number of agents. Additionally, we will ignore errors for now, and modify the BCP later to deal with them.

So let $\PP=(Q\times G,\Sigma,\delta,I,O)$ denote our BCP with global states $G$, where we set $\Sigma\defeq\{1,...,k\}$. We want to use Lemma~\ref{lem:simulatewell} to simulate the steps of the CM, so let $\PP'=(Q'\times G',\delta')$ denote a BP from that Lemma s.t.\ $\PP'$ reaches a final configuration with probability $1-\mathcal{O}(n^{-r-1})$, where $r$ is chosen s.t.\ $f\in\mathcal{O}(n^r)$.

We use local states $Q\defeq\{1_1,...,1_k\}\cup Q'$, where $C(1_i)$ represents the value of counter $i\in\{1,...,l\}$ of the CM in a configuration $C$. As global states, we have $G\defeq (G'\cup\Op{init})\times S$. Note that we have states $0,1\in Q'$ as well, with $1$ representing a “working register” in the same manner as the other counters. Agents not belonging to a counter are in state $0$. This means that the initial states are simply given by $I(i)\defeq1_i$ for $i\in\{1,...k\}$.

To perform a step of the CM, we load the affected counter into the “working register” represented by state $1\in Q'$ and let $\PP'$ run. After it has terminated, we write the value back into the original counter and move to the next state. We will now describe these transitions formally, so let $\mathcal{T}_\circ(s)=(i,j,s_0',s_1')$ denote a transition of the CM, with $\circ\in\{1,2\}$, $s\in S$. To initialise the working register, we need
\begin{align}
\begin{alignedat}{2}
&(1_i,\Op{init},s)&&\mapsto(1,j,s),\,\{1_i\mapsto 1\} \\
&(0,\Op{init},s)&&\mapsto(0,j,s),\,\{1_i\mapsto 1\}
\end{alignedat}\tag{init}
\end{align}
Then we faithfully execute the transitions of $\PP'$.
\begin{align}
(q,j,s)\mapsto(q',j',s),\,f&\qquad\text{for $\delta'(q,j)=((q',j'),f)$}\tag{run}
\end{align}
Once $\PP'$ reaches a final state, we move to the next state of $A$.
\begin{align}
\begin{alignedat}{2}
&(1,\Op{done}_b,s)\mapsto(1_i,\Op{init},s_b'),&&\,\{1\mapsto 1_i\} \\
&(0,\Op{done}_b,s)\mapsto(0,\Op{init},s_b'),&&\,\{1\mapsto 1_i\}
\end{alignedat}
&\qquad\text{for $b\in\{0,1\}$}\tag{finish}
\end{align}

Note that we have two outgoing (init) transitions, one for each $\circ$. Matching the execution of CMs, we pick one of these uniformly at random, using our construction from Section~\ref{sec:nondeterministic}. We also remark that the BP $\PP'$ is only guaranteed to simulate the CM if the counter values do not exceed $n$. This is ensured by $A$ being bounded.

We know that $A$ eventually reaches either state $0$ or $1$, which determines the result of its computation. So we define the accepting states of $\PP$ based on that as $O\defeq\{(q,j,s)\in Q\times G:s=1\}$.

As mentioned, the above disregards error conditions. When executing $\PP'$ via transition (run) it might end up in a failing configuration, meaning that at least one agent has local state $\perp$. There are no transitions which would cause an agent to leave state $\perp$ (in particular, recall that the set of failing configurations is closed under reachability for $\PP'$). Hence any agent with local state $\perp$ remains so, and we can modify $\PP$ to cause such an agent to initiate a reset. For this we have each agent remember its initial state, and state $\perp$ sends a broadcast which reverts each agent to its initial state, effectively restarting the computation.

Formally, we do this transformation (as well as simulating $l$ agents by a single one) as follows. Let $\PP^*=(Q^*,\Sigma,\delta^*,I^*,O^*)$ denote the new BCP. We use states $Q^*\defeq (Q\times G)^{l+1}$, where the first component stores the initial state and the latter $l$ components simulate $l$ agents. The input mapping is $I(i)\defeq(1_i,1_i,0,...,0)$ for $i\in\{1,...,k\}$, and the accepting states are $O^*\defeq\{(q_0,...,q_l)\in Q^*:q_1\in O\}$.

For $i\in\{1,...,l\},q\in Q\times G$ we execute a transition $\delta(q_i)=(q'_i,f)$ of $\PP$ as
\begin{align*}
&(q_0,...,q_l)\mapsto(q_0,q'_1,...,q'_l),\,\{(r_0,...,r_l)\mapsto (r_0,f(r_1),...,f(r_l))\}
\end{align*}
where $q'_m\defeq f(q_m)$ for $m\in\{1,...,l\}\setminus\{i\}$. To perform the resets, we add a transition for each $(q_0,...,q_l)\in Q^*$ with $q_i=\perp$ for some $i$.
\begin{align*}
&(q_0,...,q_l)\mapsto(q_0,q_0,0,...,0),\,\{(r_0,...,r_l)\mapsto(r_0,r_0,0,...,0)\}
\end{align*}

We execute one step of the CM every $\mathcal{O}(n\log n)$ interactions, in expectation. If no failure occurs, the BCP stabilises after $\mathcal{O}(n^{r+1}\log n)\subseteq\mathcal{O}(f(n)\,n\log n)$ steps, in expectation. A single step can fail with probability $\mathcal{O}(n^{-r-1})$, so the probability that all steps succeed is at least $1-\mathcal{O}(n^{-1})$. If a step fails, the computation will restart (after $\mathcal{O}(n)$ steps in expectation), increasing the expected running time by a factor of $1/(1-\mathcal{O}(n^{-1}))\in\mathcal{O}(1)$.
\end{proof}